\documentclass[lettersize,journal]{IEEEtran}

\usepackage[T1]{fontenc}
\usepackage{amssymb}
\usepackage[utf8]{inputenc}
\usepackage[english]{babel}
\usepackage{amsmath}			
\usepackage{courier}

\usepackage{comment}				
\usepackage[T1]{fontenc}
\usepackage{mathtools}
\usepackage{algcompatible,amsmath}
\usepackage[ruled,vlined]{algorithm2e}
\usepackage{graphicx}
\usepackage{multicol}
\usepackage{afterpage}
\usepackage[noend]{algpseudocode}
\usepackage{color}
\usepackage{etoolbox}
\usepackage{url}
\usepackage{siunitx,booktabs}
\usepackage{placeins}
\usepackage{bbm}
\usepackage{soul}
\usepackage[noadjust]{cite}

\usepackage{amsfonts}

\usepackage{accents}

\normalsize
\makeatletter
\def\underbracex#1#2{\mathop{\vtop{\m@th\ialign{##\crcr
				$\hfil\displaystyle{#2}\hfil$\crcr
				\noalign{\kern3\p@\nointerlineskip}%
				#1\crcr\noalign{\kern3\p@}}}}\limits}

\def\upbracefilla{$\m@th \setbox\z@\hbox{$\braceld$}%
	\bracelu\leaders\vrule \@height\ht\z@ \@depth\z@\hfill 
	\kern\p@\vrule \@width\p@\kern\p@\vrule \@width\p@\kern\p@\vrule \@width\p@
	$}

\def\upbracefillb{$\m@th \setbox\z@\hbox{$\braceld$}%
	\vrule \@width\p@\kern\p@\vrule \@width\p@\kern\p@\vrule \@width\p@\kern\p@
	\leaders\vrule \@height\ht\z@ \@depth\z@\hfill\bracerd
	\braceld\leaders\vrule \@height\ht\z@ \@depth\z@\hfill
	\kern\p@\vrule \@width\p@\kern\p@\vrule \@width\p@\kern\p@\vrule \@width\p@
	$}

\def\upbracefillc{$\m@th \setbox\z@\hbox{$\braceld$}%
	\vrule \@width\p@\kern\p@\vrule \@width\p@\kern\p@\vrule \@width\p@\kern\p@
	\leaders\vrule \@height\ht\z@ \@depth\z@\hfill
	\kern\p@\vrule \@width\p@\kern\p@\vrule \@width\p@\kern\p@\vrule \@width\p@
	$}

\def\upbracefilld{$\m@th \setbox\z@\hbox{$\braceld$}%
	\vrule \@width\p@\kern\p@\vrule \@width\p@\kern\p@\vrule \@width\p@\kern\p@
	\leaders\vrule \@height\ht\z@ \@depth\z@\hfill\braceru$}
	
\def\underbracex#1#2{\mathop{\vtop{\m@th\ialign{##\crcr
				$\hfil\displaystyle{#2}\hfil$\crcr
				\noalign{\kern3\p@\nointerlineskip}%
				#1\crcr\noalign{\kern3\p@}}}}\limits}

\def\upbracefilla{$\m@th \setbox\z@\hbox{$\braceld$}%
	\bracelu\leaders\vrule \@height\ht\z@ \@depth\z@\hfill 
	\kern\p@\vrule \@width\p@\kern\p@\vrule \@width\p@\kern\p@\vrule \@width\p@
	$}

\def\upbracefillb{$\m@th \setbox\z@\hbox{$\braceld$}%
	\vrule \@width\p@\kern\p@\vrule \@width\p@\kern\p@\vrule \@width\p@\kern\p@
	\leaders\vrule \@height\ht\z@ \@depth\z@\hfill\bracerd
	\braceld\leaders\vrule \@height\ht\z@ \@depth\z@\hfill
	\kern\p@\vrule \@width\p@\kern\p@\vrule \@width\p@\kern\p@\vrule \@width\p@
	$}

\def\upbracefillbd{$\m@th \setbox\z@\hbox{$\braceld$}%
	\vrule \@width\p@\kern\p@\vrule \@width\p@\kern\p@\vrule \@width\p@\kern\p@
	\bracerd\braceld
	\leaders\vrule \@height\ht\z@ \@depth\z@\hfill\braceru$}

\makeatletter
\patchcmd{\@makecaption}
{\scshape}
{}
{} 
{}
\makeatother
\usepackage{subcaption}
\usepackage{graphicx}
\usepackage{amsthm}

\newtheorem{proposition}{Proposition}

\newtheorem*{proof*}{Proof}

\newtheorem{remark}{Remark}				

\newcommand{\norm}[1]{\left\lVert#1\right\rVert}

\IEEEoverridecommandlockouts
\makeatletter
\let\origIEEEPARstart\IEEEPARstart
\renewcommand{\IEEEPARstart}[3][1.1]{%
	\def\@IEEEPARstartDROPDEPTH{#1\baselineskip}%
	\origIEEEPARstart{#2}{#3}%
}

\def\endthebibliography{%
	\def\@noitemerr{\@latex@warning{Empty `thebibliography' environment}}%
	\endlist
}
\makeatother

\ifCLASSINFOpdf

\else

\fi

\begin{document}
	
	\title{Joint Spatial and Spectral Hybrid Precoding for Multi-User MIMO-OFDM Systems}

\author{Navid~Reyhanian,
	Reza Ghaderi Zefreh, Parisa Ramezani, and Emil Bj{\"o}rnson,
	\IEEEmembership{Fellow, IEEE}
	\thanks{This work was presented in part for presentation at the IEEE International Conference on Communications (ICC), Glasgow, UK, May 24--28, 2026.}
	\thanks{N. Reyhanian was with the Department of Electrical Engineering, University of Minnesota Twin-Cities, Minneapolis, MN, USA 55455 during this research. He is now with Cisco Systems, Milpitas, CA, USA 95035 (email: navid@umn.edu).}
	\thanks{R.G. Zefreh is Independent Researcher, Shiraz, Fars, Iran (email: r.ghaderi@alumni.iut.ac.ir).}
	\thanks{P. Ramezani and E. Bj{\"o}rnson are with the Department of Computer Science, KTH Royal Institute of Technology, Stockholm, Sweden (email: \{parram,emilbjo\}@kth.se).}\vspace{-0.7cm}
}


\allowdisplaybreaks

\maketitle

\begin{abstract}
	Millimeter wave (mmWave) multiple-input multiple-output (MIMO) systems operate over wide bandwidths and frequency-selective channels, making orthogonal frequency-division multiplexing (OFDM) a natural transmission scheme. In such systems, fully digital precoding is often impractical because the large antenna arrays require high hardware cost and power consumption, so hybrid precoding that combines digital and radio frequency (RF) processing is an attractive alternative. However, OFDM introduces high signal peaks that may cause clipping and generate out-of-band (OOB) emissions, while practical, nonideal phase shifters (PSs) at the RF precoder and user combiner suffer from phase errors. 	We study the problem of robust digital-RF precoding optimization for the downlink sum-rate maximization in multi-user (MU) MIMO-OFDM systems under maximum transmit power, clipping, and OOB emission mask constraints. The formulated maximization problem is nonconvex and difficult to solve. We propose a weighted minimum mean squared error (WMMSE) based block coordinate descent (BCD) method to iteratively optimize digital-RF precoders at the transmitter and digital-RF combiners at the users. Low-cost and scalable optimization approaches are proposed to efficiently solve the BCD subproblems. Extensive simulation results are conducted to demonstrate the efficiency of the proposed approaches and exhibit their superiority relative to well-known benchmarks.
	
\end{abstract}
\begin{IEEEkeywords}
	Millimeter wave technology,	hybrid precoding, MU-MIMO-OFDM, phase errors, block coordinate descent, alternating direction method of multipliers.
\end{IEEEkeywords}

\IEEEpeerreviewmaketitle

\section{Introduction}
\label{sec:introduction}
\IEEEPARstart{M}{illimeter} wave (mmWave) communication is a promising solution for supporting high-data-rate wireless systems due to the large available bandwidth \cite{7913599}. Although mmWave links suffer from severe propagation loss, the small wavelength enables large antenna arrays at the transmitter and receivers. This makes massive MIMO particularly attractive for mmWave multi-user MIMO (MU-MIMO), where spatial multiplexing and precoding are used to serve multiple users on the same time-frequency resources while controlling inter-user interference.

Fully-digital precoding, which requires a separate radio frequency (RF) chain for each antenna, introduces impractical hardware complexity and excessive power consumption in large-scale arrays. To address this, a hybrid phase-shifter (PS) based digital-RF precoding architecture is adopted, combining a high-dimensional RF beamformer implemented via PSs with a low-dimensional digital beamformer \cite{7913599,7397861,6928432}. This approach strikes a balance between performance and hardware efficiency, allowing hybrid precoding to approach the performance of fully-digital schemes \cite{8048606}.  

However, mmWave systems are intended to operate over wideband channels, where frequency selectivity poses a new challenge. Hybrid precoding in such frequency-selective channels, typically implemented with orthogonal frequency division multiplexing (OFDM), requires designing a shared RF beamformer across all subcarriers, while the digital beamformer is optimized on a per-subcarrier basis \cite{9467491}. This differs from flat-fading channel designs and introduces additional complexity in maintaining performance across multiple subcarriers. 

Additionally, OFDM introduces practical impairments, including high-amplitude peaks that can cause clipping \cite{7366560}. Such peaks impose stringent requirements on digital-to-analog converters (DACs) and power amplifiers (PAs), and may induce signal distortion and spectral regrowth \cite{9151131}, degrading both in-band and adjacent-channel performance, which is especially critical in wideband mmWave systems. Hence, future mmWave MIMO-OFDM systems require hybrid precoding designs with peak-suppression and clipping-mitigation capability. OFDM also produces high out-of-band (OOB) radiation, particularly in mmWave MIMO settings, which can interfere with neighboring systems and degrade communication quality \cite{10048700}; this issue is more severe at base stations (BSs), where higher average transmit power is required to ensure sufficient coverage. To limit adjacent-channel interference, wireless standards impose strict OOB emission requirements, commonly assessed through the adjacent channel leakage ratio (ACLR) and the spectral emission mask \cite{9390405}.

\subsection{Related Work}

Hybrid precoding for partially- and fully-connected systems has been widely studied in \cite{7913599,7397861,9467491,10839437,10972245,10778226}. In partially-connected architectures, each antenna is connected to one RF chain, which reduces implementation complexity, whereas in fully-connected architectures, each antenna is connected to all RF chains, which improves flexibility at the cost of higher complexity and insertion loss. For single-user (SU) MIMO-OFDM systems, \cite{7913599} shows that the fully-connected architecture can achieve spectral efficiency close to that of fully digital precoding. Since RF-precoder design is nonconvex because of the unit-modulus phase-shift constraints, Riemannian conjugate gradient (RCG) methods are commonly used for RF precoder optimization, e.g., \cite{7397861,9467491}.

Various methods have been developed to mitigate high signal peaks \cite{5628256,4014335,1214054,4786507}. In \cite{5628256}, the OFDM signal amplitude is imposed as a constraint to prevent clipping. For SU single-antenna systems, \cite{4014335} proposes a constrained clipping approach that reduces in-band distortion while satisfying OOB emission limits, although signal windowing is not considered. Other studies, such as \cite{1214054,4786507}, focus on reconstructing clipped signals and equalizing OFDM clipping noise.

Spectral precoding is an effective technique to suppress OOB emissions. A least-squares notching spectral precoder to suppress OOB emissions at specific frequencies is proposed in \cite{van2009sculpting} that minimizes the Euclidean distance between the spectrally precoded and original data vectors. This approach aims to null OOB emissions at specified frequencies while preserving the integrity of the in-band signal. Mask-compliant spectral precoder schemes have been proposed in \cite{6459499,7485853,9214877}, which adhere to predefined spectral emission upper limits instead of nulling OOB at specific frequency points. The joint suppression of peak-to-average power ratio (PAPR) and OOB radiation for SU-MIMO-OFDM systems is explored in \cite{10048700}, where OOB emissions are minimized. 

Transmit signals are often precoded using separately designed spectral and spatial precoders \cite{van2009sculpting,9390405}. In high-rate MIMO-OFDM systems, however, spectral precoding can distort the spatially precoded signals, increase interference, and reduce spectral efficiency. Joint spatial and spectral precoding for MU-MIMO-OFDM is studied in \cite{taheri2020joint} using zero-forcing (ZF) and maximum-ratio transmission (MRT), but ZF performs poorly at low signal-to-noise ratio (SNR), and MRT is ineffective at suppressing multi-user interference. Minimum mean square error (MMSE) \cite{1468466} and weighted MMSE (WMMSE) \cite{5756489} better balance noise and interference, making them suitable for practical MU-MIMO systems \cite{8048606}.

PS impairments are unavoidable hardware limitations in practical mmWave hybrid precoding systems, arising from manufacturing imperfections, switching transients, low-quality RF components, and device aging. This leads to random phase and gain errors that can severely degrade system performance. These impairments are typically modeled as random deviations from ideal PS responses, with phase errors following Gaussian distributions \cite{7956180,7948784,icc2026,9767793,8653302}. Recent works have addressed these challenges through different approaches: Li \textit{et al.} \cite{9767793} propose robust hybrid precoding design using alternating optimization with outage probability constraints under PS error uncertainty, while \cite{icc2026} develops a novel method that estimates the downlink equivalent channel directly (which inherently captures the effects of imperfect PSs) and designs digital precoders based on this estimated channel rather than relying on channel reciprocity, thereby achieving robustness while reducing training overhead. 

While the aforementioned topics have been examined individually, they have not been addressed together within a unified, robust, clipping-aware, joint spatial and spectral precoding framework for hybrid partially-connected MU-MIMO-OFDM systems. Filling this gap is the main focus of this paper.

\subsection{Our Contributions}
In this paper, we study a robust hybrid precoding design for MU-MIMO-OFDM systems with partially-connected transmitter structures under clipping and OOB emission constraints. The key contributions of this paper are as follows:

\begin{itemize}
	\item
	Since spectral-mask requirements are defined through average spectral quantities over the transmitter ON period and the measurement bandwidth \cite{3GPP_TS_38_104}, we adopt block-level precoding (similar to \cite{7397861,9467491,wang2018hybrid}), where the digital and RF precoders are held fixed across the OFDM symbols of a scheduling block within the channel coherence time. Based on this, we formulate a user sum-rate maximization problem via the joint design of the transmitter digital and RF precoders and the users' digital and RF combiners. The resulting problem is nonconvex in all optimization variables, so we develop a WMMSE-based block coordinate descent (BCD) algorithm that alternately updates the variable blocks.
	
	\item
	With block-level precoding, the digital-precoder subproblem has a convex objective and three convex constraints associated with the maximum transmit-power budget, the spectral masks at specific frequencies, and the signal-amplitude clipping limit. Because the precoder is fixed over the block and the symbols are i.i.d.\ Gaussian (Section~\ref{sec:system}), each chance constraint is enforced on a representative symbol, which avoids many symbol-level constraints while remaining consistent with the block-averaged regulatory mask semantics. To solve the resulting large-dimensional subproblem efficiently, we decompose it and develop a low-cost scalable alternating direction method of multipliers (ADMM) method that sequentially solves three smaller problems, each admitting either a closed-form solution or a bisection search.
	
	\item
	The subproblems with respect to partially-connected RF precoder and fully (or partially)-connected RF combiners at the transmitter and users, respectively, are nonconvex, although the objective function is convex with respect to these variable blocks. This nonconvexity is due to the fact that there is one unit modulus constraint for each PS. 
	To optimize the PS phase shifts at the transmitter and users, we deploy a coordinate descent (CD) method. We keep all PS values fixed except one that is to be updated. For each PS, we derive a closed-form update rule, which ensures simplicity and scalability, unlike well-known iterative methods. In addition, we consider the scenario where PSs are impaired and exact phase shift adjustments are hard to achieve. Robust schemes for PS value optimizations with closed-form solutions at the transmitter and users are proposed that maximize the expected sum-rate in the presence of random phase errors. Finally, closed-form solutions are derived for digital combiners at the users.
	
	\item
	In the proposed BCD algorithm, the above blocks are optimized alternatively with the Gauss-Seidel rule. The convergence of the proposed inner subproblem solvers is theoretically supported, and the overall BCD approach is shown to yield a monotonically convergent regularized WMMSE objective value under exact inner solves. Extensive simulation results are given to demonstrate the efficiency of the proposed algorithms in reducing OOB emissions and enhancing achievable rates relative to well-known existing methods.
	
\end{itemize}
The rest of this paper is organized as follows. We explain the system model in Section \ref{sec:system} and then present the problem formulation in Section \ref{sec:probform}, where the constraints are also reformulated to manage inherent complexities.
A detailed solution illustration, algorithmic decomposition, and our corresponding optimization methods are described in Section \ref{sec:alg}. Extensive simulation results are provided and discussed in Section \ref{sec:sim}. We conclude this paper in Section \ref{sec:con}.

\section{System Model}	\label{sec:system}
In this section, we define the main system components.

\subsubsection{Hybrid RF-Digital Systems}
We consider the downlink of a MIMO-OFDM system where a BS  is equipped with $N_t$ transmit antennas and serves $K$ users over $S$ subcarriers. Let $\mathcal{S} = \{0,\ldots,S-1\}$ be the set of all subcarriers. Each user is equipped with \(N_r\) receive antennas and is scheduled on all subcarriers. We consider a partially-connected architecture where each antenna is connected to one RF chain. Here, $N_{\text{RF}}$ represents the number of RF chains. The $N_t$ antennas are split into $N_{\text{RF}}$ disjoint subarrays of equal size and $N_t/N_{\text{RF}}\in \mathbb{Z}$.  Let $\mathbf{V}_{k}^s\in\mathbb{C}^{N_{\text{RF}} \times n_k }$ denote the digital precoding matrix that the BS uses to transmit $n_k$ data streams to user $k$ per subcarrier. Let $\boldsymbol{\omega}_{k}^s \in \mathbb{C}^{n_k \times 1}$ represent the transmitted signal for user $k$ on subcarrier $s$. Each scheduling block within the channel coherence time contains multiple OFDM symbols over which the digital and RF precoders are held fixed. Consistent with normalized QAM signaling, each transmitted stream has unit average power. To enable a single precoder design per coherence block rather than per symbol, we adopt the standard circularly symmetric complex Gaussian surrogate $\boldsymbol{\omega}_{k}^s \sim \mathcal{CN}(\mathbf{0},\mathbf{I}_{n_k})$, which maximizes entropy, captures the second-order statistics of dense QAM constellations, and yields tractable block-level distributions for both the achievable rate and the per-antenna transmit waveform. The symbols are i.i.d.\ across users, subcarriers, and OFDM symbols of the block, with $\mathbb{E}[\boldsymbol{\omega}_{k}^s(\boldsymbol{\omega}_{j}^{s})^H]=\mathbf{0}$ for $j\neq k$.

The BS is equipped with both digital and RF precoders. The received signal $\mathbf{y}_{k}^s \in \mathbb{C}^{N_r \times 1}$ at user $k$ on subcarrier $s$ can be written as
\begin{equation}
	\mathbf{y}_k^s = \mathbf{H}_k^s\mathbf{V}_{\text{RF}} \sum_{j=1}^K \mathbf{V}_j^s \boldsymbol{\omega}_j^s + \mathbf{n}_k^s,\nonumber
\end{equation}
where $\mathbf{H}_k^s \in \mathbb{C}^{N_r \times N_t}$ represents the channel matrix from the BS to user $k$ on subcarrier $s$, and $\mathbf{n}_k^s \in \mathbb{C}^{N_r \times 1}$ denotes the additive white complex Gaussian noise with distribution $\mathcal{CN}(\mathbf{0}, \sigma_{\text{noise},k}^{s\,2}\mathbf{I}_{N_r})$.
Moreover,  $\mathbf{V}_{\text{RF}}\in\mathbb{C}^{N_t \times N_{\text{RF}} }$ denotes the applied phase shift matrix of PSs. Due to partially-connected architecture, each row of $\mathbf{V}_{\text{RF}}$ has exactly one non-zero complex element whose absolute value is equal to one. If the $a^\text{th}$ antenna is connected to the $m_a^\text{th}$ RF chain, $\mathbf{V}_{\text{RF}}[a,m_a]=e^{\jmath \varphi}$ where $\jmath=\sqrt{-1}$ and $\varphi$ is the applied phase rotation.

In addition to the transmitter, each user is equipped with an RF combining matrix. The decoded signal vector on subcarrier $s$ at the $k^\text{th}$ user can be expressed as $
\hat{\boldsymbol{\omega}}_k^s = \mathbf{U}_k^{s^H}\mathbf{U}_{\text{RF},k}^{H} \mathbf{y}_k^s,$
where $\mathbf{U}_{k}^s\in\mathbb{C}^{N_{\text{RF},k} \times n_k }$ and $\mathbf{U}_{\text{RF},k}\in\mathbb{C}^{N_r \times N_{\text{RF},k} }$ are the digital and RF combining matrices, respectively.  Let \(\{\mathcal A_m\}_{m=1}^{N_{\rm RF}}\) be the disjoint transmit-antenna partition, where \(\mathcal A_m\subseteq\{1,\ldots,N_t\}\) and \(|\mathcal A_m|=N_t/N_{\rm RF}\). The feasible set of the partially-connected RF precoder is \(\mathcal V_{\rm RF}\triangleq\{\mathbf V\in\mathbb C^{N_t\times N_{\rm RF}}:\mathbf V[a,m]=0\ \text{for }a\notin\mathcal A_m,\ |\mathbf V[a,m]|=1\ \text{for }a\in\mathcal A_m\}\). For the RF combiner of user \(k\), let \(\mathcal E_k\subseteq\{1,\ldots,N_r\}\times\{1,\ldots,N_{{\rm RF},k}\}\) denote the feasible nonzero entries. This set allows either a partially- or fully-connected receiver architecture. The corresponding feasible set is \(\mathcal U_{{\rm RF},k}\triangleq\{\mathbf U\in\mathbb C^{N_r\times N_{{\rm RF},k}}:\mathbf U[a,m]=0\ \text{for }(a,m)\notin\mathcal E_k,\ |\mathbf U[a,m]|=1\ \text{for }(a,m)\in\mathcal E_k\}\). We assume each \(\mathbf U_{{\rm RF},k}\in\mathcal U_{{\rm RF},k}\) has full column rank.

\subsubsection{Power Constraint for Hybrid Precoding}
The power budget of the BS is constrained as
\begin{align}
	& \quad \sum_{k=1}^K \text{tr}(\mathbf{V}_{k}^{s^H}\mathbf{V}_{\text{RF}}^H\mathbf{V}_{\text{RF}}\mathbf{V}_{k}^s) \leq P^s, \quad \forall s.\label{eq:rfpowerbudget}
\end{align}
where $\text{tr}(\cdot)$ denotes the trace operator and $P^s$ is the total power budget of the BS on subcarrier $s$.

\subsubsection{Clipping Constraint}

Consider an oversampling factor $\ell$. For the $a^\text{th}$ antenna, the cyclic-prefix (CP) inclusive OFDM symbol is observed in discrete time over
$n\in\{-\ell N_{\mathrm{CP}},-\ell N_{\mathrm{CP}}+1,\ldots,\ell S-1\}$. Based on this, the transmitted waveform can be expressed as \cite{van2009sculpting}
\begin{equation}
	\tilde{x}_{\text{CP}}^a[n] = \mathbf{p}^T[n]\mathbf{g}^a
	= \sum_{s\in \mathcal{S}} p^{s}[n] \mathbf{g}^a[s],\label{eq:ofdm_sym}
\end{equation}
where $\mathbf{p}[n] = [p^{0}[n], \ldots, p^{S-1}[n]]^T$ contains the sampled subcarrier pulses on the $\ell S$-point grid, and $\mathbf{g}^a\in\mathbb{C}^{S}$ collects the precoded frequency-domain symbols associated with antenna $a$. In the hybrid architecture, this vector is written as
\begin{equation}
	\begin{split}
		\mathbf g^a
		&=
		\begin{bmatrix}
			\mathbf V_{\mathrm{RF}}[a,:]\sum_{k=1}^{K}\mathbf V_k^{0}\boldsymbol{\omega}_k^{0}\\
			\vdots\\
			\mathbf V_{\mathrm{RF}}[a,:]\sum_{k=1}^{K}\mathbf V_k^{S-1}\boldsymbol{\omega}_k^{S-1}
		\end{bmatrix} = \mathbf V_{\mathrm{RF}}[a,m_a]\mathbf w^a,
	\end{split}
\end{equation}
where $m_a$ denotes the unique RF-chain index connected to antenna $a$, i.e., $\mathbf{V}_\text{RF}[a,m_a]\neq 0$, and
\begin{align}
	\mathbf w^a[s]=\sum_{k=1}^{K}\mathbf{V}_{k}^{s}[m_a,:]\boldsymbol{\omega}_k^{s},  \forall a,\ \forall s.\label{eq:wa_def}
\end{align}

For the $s^\text{th}$ subcarrier, the corresponding discrete-time OFDM pulse is given by \cite{van2009sculpting,9214877,9390405}
\begin{equation}
	p^s[n] = \frac{1}{\sqrt{\ell S}}\,e^{\jmath 2\pi \frac{s}{\ell S} n}\, I[n], \:\:\: s\in \mathcal{S},   \label{eq:timedomain}
\end{equation}
where $I[n]$ equals $1$ for $-\ell N_{\mathrm{CP}} \leq n \leq \ell S-1$ and $0$ otherwise, and $N_{\mathrm{CP}}$ is the CP length in samples. We use a rectangular window as it is known to produce slow spectral decay and therefore provides a conservative baseline among commonly used window choices.

Now fix $a\in\{1,\ldots,N_t\}$ and work with an $\ell S$-point DFT grid. Define the oversampled IDFT matrix
$\mathbf F_{\ell S}^{H}\in\mathbb C^{\ell S\times S}$ with entries $
\mathbf F_{\ell S}^{H}[n,s]
\triangleq \frac{1}{\sqrt{\ell S}}e^{\jmath 2\pi \frac{ns}{\ell S}}, n\in\{0,1,\ldots,\ell S-1\}, s\in\mathcal{S},$
so that the useful time-domain OFDM block without CP, denoted by $\tilde{\mathbf x}^{a}\in\mathbb C^{\ell S}$, becomes $
\tilde{\mathbf x}^{a}=\mathbf F_{\ell S}^{H}\mathbf g^{a}
= \mathbf{V}_\text{RF}[a,m_a]\mathbf F_{\ell S}^{H}\mathbf w^{a}.$ 
To prevent clipping, the peak magnitude at antenna $a$ must satisfy 
$\|\mathbf{F}_{\ell S}^{H}\mathbf{w}^{a}\|_{\infty} \leq \chi$, where the equivalence $\|\tilde{\mathbf{x}}^a
\|_{\infty} = \|\mathbf{F}_{\ell S}^{H}\mathbf{w}^{a}\|_{\infty}$ follows 
from $|\mathbf{V}_{\text{RF}}[a,m_a]|=1$. Since the CP replicates the last 
$\ell N_{\mathrm{CP}}$ samples of the useful OFDM block, this constraint over 
$0 \leq n \leq \ell S-1$ automatically extends to the CP interval 
$n \in \{-\ell N_{\mathrm{CP}}, \ldots, -1\}$.

\subsubsection{Spectral Mask Constraint}

We next impose spectral-mask constraints to keep the OOB radiation of the transmitted waveform within the limits required by standards and regulators. In MIMO-OFDM, this issue can be severe because the effective subcarrier waveforms have sidelobes that extend energy outside the intended band. If this leakage is not controlled, adjacent channels may be affected and the transmitted signal may violate the allowed mask. For this reason, the spectrum must be shaped explicitly, usually with stronger attenuation near the band edges, so that OOB components are suppressed while spectral efficiency is preserved.

We select a fixed oversampling factor $\ell$ and consider the CP-inclusive transmit pulse $p^{s}[n]$ in \eqref{eq:timedomain} over
$n\in\{-\ell N_{\mathrm{CP}},\ldots,\ell S-1\}$. Let $F_{s,\ell}$ be the sampling rate of the oversampled discrete-time
waveform. To sample the spectrum, introduce a set of (possibly non-integer) points
$\{\gamma_i\}_{i=0}^{M-1}\subset\mathbb{R}$ measured in DFT-bin units on the 
$\ell S$ grid, where integer-valued $\gamma_i$ align with $\ell S$-point DFT-bin 
centers and non-integer $\gamma_i$ correspond to frequencies between adjacent bins. Each point $\gamma_i$
maps to the actual baseband frequency $f_i=\frac{\gamma_i}{\ell S}F_{s,\ell}$ in Hz. When $\gamma_i$ is an integer, it matches an
$\ell S$-point DFT-bin center; when it is non-integer, it corresponds to a DFT evaluation between bins.

Using the rectangular window in \eqref{eq:timedomain}, define the sampling matrix $\mathbf A\in\mathbb C^{M\times S}$ by
evaluating that formula at each $\gamma_i$, i.e.,
\begin{align}
	\mathbf A[i,s]
	=&
	\frac{1}{\sqrt{\ell S}}
	\exp\!\left(\jmath\pi\frac{\gamma_i-s}{\ell S}\,(\ell N_{\mathrm{CP}}-\ell S+1)\right)	\label{eq:A_closed}\\
	&\times
	\frac{\sin\!\left(\pi\frac{\gamma_i-s}{\ell S}(\ell S+\ell N_{\mathrm{CP}})\right)}
	{\sin\!\left(\pi\frac{\gamma_i-s}{\ell S}\right)},
	\qquad \gamma_i-s\notin \ell S\,\mathbb Z,
\nonumber
\end{align}
and $\mathbf A[i,s]=L/\sqrt{\ell S}$ when $\gamma_i-s\in \ell S\,\mathbb Z$, with $L\triangleq \ell S+\ell N_{\mathrm{CP}}$ \cite{9214877,van2009sculpting}.

For a fixed antenna $a$, \eqref{eq:ofdm_sym}--\eqref{eq:timedomain} show that the CP-inclusive spectrum at the sampled locations satisfies
\begin{align}
	&X^a(\gamma_i)
	\;\triangleq\;
	\sum_{n=-\ell N_{\mathrm{CP}}}^{\ell S-1} \tilde{x}_{\text{CP}}^a[n]\;e^{-\jmath 2\pi \frac{\gamma_i}{\ell S}n}
	\;=\;
	\sum_{s\in\mathcal S}\mathbf g^a[s]\;\mathbf A[i,s]
	\;\nonumber\\
	&=\;
	\mathbf A[i,:]\mathbf g^a,
	\qquad i=0,\ldots,M-1,
	\nonumber
\end{align}
and the associated single-symbol periodogram-type power spectral density (PSD) sample at $\gamma_i$ is
\begin{equation}
	\widehat S_{\tilde{x}_{\text{CP}}^a \tilde{x}_{\text{CP}}^a}(\gamma_i)
	\;\triangleq\;
	\frac{1}{L\,F_{s,\ell}}\bigl|X^a(\gamma_i)\bigr|^2
	\;=\;
	\frac{1}{L\,F_{s,\ell}}\bigl|\mathbf A[i,:]\mathbf g^a\bigr|^2.
	\label{eq:PSD_inst_A}
\end{equation}
In the sequel, the single-symbol spectrum sample at frequency location $\gamma$,
$\widehat S_{\tilde{x}_{\text{CP}}^a \tilde{x}_{\text{CP}}^a}(\gamma)= \frac{1}{L F_{s,\ell}}|X^a(\gamma)|^2$,
is used as a normalized indicator of spectral leakage. The emission-mask constraints are enforced directly on these samples at a finite set of mask frequencies. Under $\mathbf g^a=\mathbf V_{\text{RF}}[a,m_a]\mathbf w^a$, one has
$X^a(\gamma_i)=\mathbf V_{\text{RF}}[a,m_a]\;\mathbf A[i,:]\mathbf w^a$ and therefore
$|X^a(\gamma_i)|^2 = |\mathbf V_{\text{RF}}[a,m_a]|^2\,|\mathbf A[i,:]\mathbf w^a|^2$; in particular, if
$|\mathbf V_{\text{RF}}[a,m_a]|=1$, then the PSD samples in \eqref{eq:PSD_inst_A} are unaffected by the analog precoder's magnitude. Hence, without loss of generality, we define $\mathbf x^a \triangleq \mathbf F_{\ell S}^{H}\mathbf w^a$ and impose clipping constraint using $\mathbf x^a$, since $\| \tilde{\mathbf x}^a \|_{\infty}=\| \mathbf{x}^a \|_{\infty}$.

The mask itself is enforced only at a finite set of mask frequencies
$\{f_1,\ldots,f_G\}$. These points are different from the generic sampling locations
$\{\gamma_i\}_{i=0}^{M-1}$ used for PSD evaluation or plotting. The set $\{\gamma_i\}$ may be taken dense over a broad frequency range to inspect the spectrum, while $\{f_j\}$ is only the design set used to impose compliance. Following \cite{van2009sculpting}, the mask points are placed densely near the band edges and around interference regions. Sampling more densely in those critical regions gives a practical discrete approximation of continuous mask constraints and still lets the constraint locations adapt to time-varying interference.

Map each physical mask frequency $f_j$ to its DFT-bin location on the $\ell S$ grid through
$\gamma_j \triangleq \frac{\ell S}{F_{s,\ell}}f_j$. Let $\mathbf A_n\in\mathbb C^{G\times S}$ denote the matrix
obtained by evaluating \eqref{eq:A_closed} at $\{\gamma_j\}_{j=1}^{G}$, or, equivalently, by selecting the corresponding rows of
$\mathbf A$ when $\{\gamma_j\}\subset\{\gamma_i\}$. Under hybrid precoding, $
X^a(\gamma_j)=\mathbf V_{\text{RF}}[a,m_a]\;\mathbf A_n[j,:]\mathbf w^a,
j=1,\ldots,G,$
and the mask constraints at antenna $a$ are imposed as
\begin{equation}
	\frac{\bigl|\mathbf A_n\mathbf w^a\bigr|^{\circ 2}}{L\,F_{s,\ell}} \preceq \frac{\mathbf r}{L\,F_{s,\ell}},
	\qquad \mathbf r=[\mathbf r[1],\ldots,\mathbf r[G]]^T, \forall a,
	\label{eq:mask_constraint}
\end{equation}
where $\mathbf r[j] \triangleq L\,F_{s,\ell}\,S_{\text{max}}(f_j)$, $S_{\text{max}}(f_j)$ is the maximum allowable PSD. In other words, \eqref{eq:mask_constraint} requires choosing the
digital precoders $\mathbf V_k^s$ so that the induced per-antenna vector $\mathbf w^a$ produces OFDM symbols whose spectral
samples remain within the prescribed limits at $\{f_1,\ldots,f_G\}$.

The next section recasts the mask and clipping requirements as chance constraints on periodogram samples and per-antenna peaks. Under the i.i.d.\ Gaussian block model, the sufficient conditions in Section~\ref{sec:reform} also bound $\mathbb{E}[|\mathbf A_n[j,:]\mathbf w^a|^2]$ and the per-sample variance of $\mathbf x^a$, providing a reliable way to meet block-averaged 3GPP NR ACLR and PSD compliance requirements \cite{9214877}.

\section{Sum-Rate Maximization Problem}
\label{sec:probform}
The goal in this paper is to maximize the sum-rate of users across all subcarriers when each user receives data on all subcarriers. The sum-rate is $\sum_{k=1}^K \sum_{s\in \mathcal{S}} R^s_{k}$ where $R^s_{k}$ is the achievable rate of user $k$ on subcarrier $s$. The rate calculation assumes that users have perfect channel state information (CSI) and treat interference from other users as independent additive noise, which is a standard assumption in MU-MIMO systems that enables tractable analysis while providing simple achievable rates. Under these assumptions, $R^s_{k}$ is calculated as follows \cite{9467491,9174747}:
\begin{align}
	R^s_{k} = \log_2\Big|\mathbf{I}_{n_k} +&\mathbf{U}_k^{s^H} \mathbf{U}_{\text{RF},k}^H\mathbf{H}^s_{k} \mathbf{V}_{\text{RF}}\mathbf{V}^s_{k}\nonumber\\
	& \times {\mathbf{V}^s_{k}}^H \mathbf{V}_{\text{RF}}^H{\mathbf{H}^s_{k}}^H\mathbf{U}_{\text{RF},k}\mathbf{U}_k^s\mathbf{J}_k^{s^{-1}} \Big|,\label{eq:rate-cal}
\end{align}
where $\mathbf{J}_k^s = \sum_{j=1,j \neq k}^K \mathbf{U}_k^{s^H}\mathbf{U}_{\text{RF},k}^H \mathbf{H}^s_{k} \mathbf{V}_{\text{RF}}\mathbf{V}^s_{j} {\mathbf{V}^s_{j}}^H \mathbf{V}_{\text{RF}}^H{\mathbf{H}^s_{k}}^H$\\$\times \mathbf{U}_{\text{RF},k}\mathbf{U}_k^s+ \sigma_{\text{noise},k}^{s\,2} \mathbf{U}_k^{s^H}\mathbf{U}_{\text{RF},k}^H\mathbf{U}_{\text{RF},k}\mathbf{U}_k^s.$
The achievable rate is computed using the above matrix formulation where $\mathbf{J}_k^s$ represents the interference-plus-noise covariance matrix that accounts for MU interference and noise contributions in the MIMO system \eqref{eq:rate-cal}.

The optimization is carried out over the transceiver matrices. Under the block-level i.i.d.\ Gaussian symbol model, the per-antenna quantities \(\mathbf w^a\) and \(\mathbf x^a\) are stationary across the OFDM symbols of the block. The clipping and spectral-mask requirements are hence imposed probabilistically on a representative symbol. The sum-rate maximization problem is 
\begin{align}\label{opt:main}
	\max_{\substack{\mathbf V_{\rm RF}\in\mathcal V_{\mathrm{RF}},\,\{\mathbf U_{{\rm RF},k}\in\mathcal U_{\mathrm{RF},k}\}_{k=1}^K,\\
			\{\mathbf U_k^s,\mathbf V_k^s\}_{k=1,s\in\mathcal S}^{K}}}
	&\sum_{k=1}^K \sum_{s\in\mathcal S} R_k^s \\
	\text{s.t.}\ & \sum_{k=1}^K \norm{ \mathbf V_{\mathrm{RF}} \mathbf V_k^s}_F^2\le P^s,&& \forall s,\nonumber\\
	&\mathbb P\!\left(\|\mathbf x^a\|_\infty \le \chi\right)\ge 1-\epsilon_1,&& \forall a,\nonumber\\
	&\mathbb P\!\left(|\mathbf A_n\mathbf w^a|^{\circ 2}\preceq \mathbf r\right)\ge 1-\epsilon_2,&& \forall a.\nonumber
\end{align}
where \(\epsilon_1\in(0,1)\) and \(\epsilon_2\in(0,1)\) are violation parameters.

\subsection{Reformulation of Constraints}
\label{sec:reform}
Here, we derive deterministic sufficient conditions for the chance constraints in \eqref{opt:main} using the following propositions.

\begin{proposition}\label{prop:clip_h}
	To avoid signal clipping at antenna \(a\) with probability at least \(1-\epsilon_1\),
	when partially-connected hybrid precoding is utilized, it is sufficient that for RF chain \(m_a\),
	\begin{align}\label{eq:clipping-probablistic-bound2}
		\sqrt{\frac{\ln(\ell S/\epsilon_1)}{\ell S}\left(\sum_{k=1}^{K}\sum_{s\in \mathcal{S}}\norm{\mathbf{V}_k^s[m_a,:]}_2^2 \right)} \leq \chi.
	\end{align}
\end{proposition}

\begin{proof}
	Fix an antenna \(a\) and let \(m_a\) be its (unique) connected RF chain. Over the useful block,
	\(\mathbf x^{a}=\mathbf F_{\ell S}^{H}\mathbf w^{a}\) and \(|\mathbf V_{\mathrm{RF}}[a,m_a]|=1\) implies that \(\norm{\mathbf x^{a}}_\infty\) is
	unaffected by the RF precoder magnitude. For any \(n\in\{0,\ldots,\ell S-1\}\),
	\[
	x^{a}[n]=\frac{1}{\sqrt{\ell S}}\sum_{s\in\mathcal S} e^{\jmath 2\pi \frac{ns}{\ell S}}\, \mathbf w^{a}[s].
	\]
	Since \(\boldsymbol\omega_k^s\sim\mathcal{CN}(\mathbf 0,\mathbf I_{n_k})\) is independent across \(k\) and \(s\), each \(\mathbf w^{a}[s]\) is
	zero-mean complex Gaussian with variance
	\(\mathbb E|\mathbf w^{a}[s]|^2=\sum_{k=1}^K\|\mathbf V_k^s[m_a,:]\|_2^2\), and cross-terms vanish. Hence,
	\[
	\hspace{-.2cm}\sigma_x^2 \triangleq \mathbb E|x^{a}[n]|^2
	= \frac{1}{\ell S}\sum_{s\in\mathcal S}\mathbb E|\mathbf w^{a}[s]|^2
	= \frac{1}{\ell S}\sum_{k=1}^K\sum_{s\in\mathcal S}\|\mathbf V_k^s[m_a,:]\|_2^2,
	\]
	which is independent of \(n\). Since \(x^{a}[n]\sim\mathcal{CN}(0,\sigma_x^2)\), we have
	\(\mathbb P(|x^{a}[n]|>\chi)=\exp(-\chi^2/\sigma_x^2)\).
	By the union bound over \(n=0,\ldots,\ell S-1\),
	\[
	\mathbb P(\|\mathbf x^{a}\|_\infty>\chi)\le \sum_{n=0}^{\ell S-1}\mathbb P(|x^{a}[n]|>\chi)
	= \ell S\,\exp\!\left(-\frac{\chi^2}{\sigma_x^2}\right).
	\]
	Imposing \(\ell S\,\exp(-\chi^2/\sigma_x^2)\le \epsilon_1\) yields \(\chi^2\ge \sigma_x^2\ln(\ell S/\epsilon_1)\), which is exactly
	\eqref{eq:clipping-probablistic-bound2} for the RF chain \(m_a\). Since the CP is a repetition of samples from the useful block,
	the same bound also holds for the CP samples.
\end{proof}

The following proposition presents the probabilistic bound on the constraint \( |\mathbf A_n\mathbf w^a\bigr|^{\circ 2} \preceq \mathbf r\), which serves to enforce the spectral mask.

\begin{proposition}\label{prop:mask_hp}
	Fix an antenna \(a\) and let \(m_a\) be its unique connected RF chain. Let \(\mathbf w^a\) be defined as in Section~\ref{sec:system}. Assume \(\{\boldsymbol{\omega}_k^s\}\) are independent across \(k\) and \(s\), with \(\boldsymbol{\omega}_k^s\sim\mathcal{CN}(\mathbf 0,\mathbf I_{n_k})\). Define, for each subcarrier \(s\), \(\varsigma_{m_a}^{s\,2}\triangleq \sum_{k=1}^K \|\mathbf V_k^s[m_a,:]\|_2^2\), and for each mask index \(j\in\{1,\ldots,G\}\), \(\sigma_{j,m_a}^2 \triangleq \sum_{s\in\mathcal S} |\mathbf A_n[j,s]|^2\varsigma_{m_a}^{s\,2}\). Then,
	\begin{equation}
		\mathbb P\!\left(|\mathbf A_n\mathbf w^a|^{\circ 2} \preceq \mathbf r\right)
		\ge
		1-\sum_{j=1}^G \exp\!\left(-\frac{\mathbf r[j]}{\sigma_{j,m_a}^2}\right).
		\label{eq:mask_hp_lb}
	\end{equation}
	In particular, for any \(\epsilon_2\in(0,1)\), if for all \(j\),
	\begin{equation}
		\sum_{s\in\mathcal S} |\mathbf A_n[j,s]|^2
		\left(\sum_{k=1}^K \|\mathbf V_k^s[m_a,:]\|_2^2\right)
		\le
		\frac{\mathbf r[j]}{\ln(G/\epsilon_2)},
		\label{eq:mask_hp_suff}
	\end{equation}
	then, \(\mathbb P(|\mathbf A_n\mathbf w^a|^{\circ 2} \preceq \mathbf r)\ge 1-\epsilon_2\).
\end{proposition}

\begin{proof}
	By the assumed Gaussianity and independence, each component \(\mathbf w^a[s]\) is zero-mean complex Gaussian with variance \(\varsigma_{m_a}^{s\,2}\), and \(\mathbf w^a\) is jointly Gaussian with diagonal covariance \(\mathrm{diag}(\varsigma_{m_a}^{s\,2})_{s\in\mathcal S}\). Hence, \(\mathbf z^a\triangleq \mathbf A_n\mathbf w^a\) is jointly complex Gaussian, and each scalar \(\mathbf z^a[j]=\mathbf A_n[j,:]\mathbf w^a\) satisfies \(\mathbf z^a[j]\sim\mathcal{CN}(0,\sigma_{j,m_a}^2)\).
	
	For \(z\sim\mathcal{CN}(0,\sigma^2)\), \(\mathbb P(|z|^2>t)=\exp(-t/\sigma^2)\). Therefore, \(\mathbb P(|\mathbf z^a[j]|^2>\mathbf r[j])=\exp(-\mathbf r[j]/\sigma_{j,m_a}^2)\). By the union bound,
	\[
	\mathbb P\!\left(|\mathbf z^a|^{\circ 2}\preceq \mathbf r\right)
	\ge
	1-\sum_{j=1}^G \exp\!\left(-\frac{\mathbf r[j]}{\sigma_{j,m_a}^2}\right),
	\]
	which proves \eqref{eq:mask_hp_lb}. If \eqref{eq:mask_hp_suff} holds, then \(\sigma_{j,m_a}^2\le \mathbf r[j]/\ln(G/\epsilon_2)\), so \(\exp(-\mathbf r[j]/\sigma_{j,m_a}^2)\le \epsilon_2/G\). Hence, the sum is at most \(\epsilon_2\), which proves the result.
\end{proof}

\begin{remark}
	The chance constraints are per-symbol and per-antenna. By stationarity, they bound the expected fraction of violating symbols in the block by \(\epsilon_i\) for each antenna. Proposition~\ref{prop:mask_hp} also bounds \(\mathbb E[|\mathbf A_n[j,:]\mathbf w^a|^2]\), and therefore, by the linearity of expectation, the expected block-averaged sampled PSD at each mask frequency. A joint \(1-\epsilon_i\) guarantee over all \(B\) OFDM symbols for a fixed antenna follows by replacing \(\epsilon_i\) with \(\epsilon_i/B\); over all \(BN_t\) symbol and antenna pairs, by replacing \(\epsilon_i\) with \(\epsilon_i/(BN_t)\). These changes only rescale the logarithmic constants in Propositions~\ref{prop:clip_h} and \ref{prop:mask_hp} and leave the algorithmic structure unchanged.
\end{remark}

\section{The Proposed WMMSE-Based Algorithm}
\label{sec:alg}
In this section, we develop a novel algorithm based on the WMMSE approach that uses the relationship between the total sum-rate and the mean-square error (MSE). In the WMMSE-based approach, a weighted MSE is minimized with respect to the aforementioned blocks of variables, rather than maximizing \eqref{opt:main}. First, we compute the MSE matrix of the estimated signal for user $k$  on subcarrier $s$ as 
\begin{align}
	&\mathbf{E}_{k}^s =\mathbb{E}_{\boldsymbol{\omega},\mathbf{n}}[(\hat{\boldsymbol{\omega}}_k^s-\boldsymbol{\omega}_k^s)(\hat{\boldsymbol{\omega}}_k^s-\boldsymbol{\omega}_k^s)^H]\nonumber\\
	&=(\mathbf{I}_{n_k} - \mathbf{U}_k^{s^H}\mathbf{U}_{\text{RF},k}^{H} \mathbf{H}_k^s\mathbf{V}_{\text{RF}} \mathbf{V}_k^s)\nonumber\\
	&(\mathbf{I}_{n_k} - \mathbf{U}_k^{s^H} \mathbf{U}_{\text{RF},k}^{H}\mathbf{H}_k^s\mathbf{V}_{\text{RF}}\mathbf{V}_k^s)^H \nonumber\\
	&+ \mathbf{U}_k^{s^H}\mathbf{U}_{\text{RF},k}^{H} \Big( \sum_{j =1, j\neq k}^K \mathbf{H}_k^s\mathbf{V}_{\text{RF}}\mathbf{V}_j^s \mathbf{V}_j^{s^H} \mathbf{V}_{\text{RF}}^H \mathbf{H}_k^{s^H}  \Big)\mathbf{U}_{\text{RF},k} \mathbf{U}_k^s\nonumber\\
	& +\sigma_{\text{noise},k}^{s\,2}\mathbf{U}_k^{s^H}\mathbf{U}_{\text{RF},k}^H\mathbf{U}_{\text{RF},k}\mathbf{U}_k^s.\label{eq:E}
\end{align}
The expectation is taken with respect to $\{\boldsymbol{\omega}_k^s\}_{k=1}^K$ and $\{\mathbf{n}_k^s\}_{k=1}^K$. In deriving this expectation, we utilized the properties $\mathbb{E}[\boldsymbol{\omega}_k^s\boldsymbol{\omega}_k^{s^H}]=\mathbf{I}_{n_k}$ and $\mathbb{E}[\boldsymbol{\omega}_k^s\boldsymbol{\omega}_l^{s^H}]=\mathbf{0}_{n_k}$ for $l\neq k$. Following the WMMSE approach \cite{5756489}, we introduce a Hermitian positive definite weight matrix \(\mathbf W_k^s\succ\mathbf 0\) applied to \(\mathbf E_k^s\). This leads to the formulation of an alternative problem to \eqref{opt:main}:
\begin{align}
	\hspace{-0.3cm}	\min_{\substack{\mathbf V_{\rm RF}\in\mathcal V_{\mathrm{RF}},\,\{\mathbf U_{{\rm RF},k}\in\mathcal U_{\mathrm{RF},k}\}_{k=1}^K,\\
			\{\mathbf U_k^s,\mathbf V_k^s,\mathbf W_k^s\succ\mathbf 0\}_{k=1,s\in\mathcal S}^{K}}}
	\quad &\hspace{-0.4cm}
	\sum_{k=1}^{K}\sum_{s\in\mathcal S}
	\Big(\operatorname{tr}(\mathbf W_k^s\mathbf E_k^s)-\log\det(\mathbf W_k^s)\Big) \nonumber\\
	\textrm{s.t.}\quad &
	\eqref{eq:rfpowerbudget},\eqref{eq:clipping-probablistic-bound2},\eqref{eq:mask_hp_suff}.
	\label{opt:main-revised}
\end{align}
To solve \eqref{opt:main-revised}, we employ a BCD approach. This method alternately minimizes \eqref{opt:main-revised} with respect to different variable blocks until the objective value converges or a prescribed stopping criterion is met.

\subsection{The Subproblem with Respect to $\mathbf{U}_k^s$}
The problem with respect to $\mathbf{U}_k^s$ is unconstrained. The optimal solution for $\mathbf{U}_k^s$ is determined by the first order optimality condition. Hence, each $\mathbf{U}_k^s$ iterate is updated as 
\begin{align}
	\mathbf{U}_k^s = &\Big(\mathbf{U}_{\text{RF},k}^H \Big(\mathbf{H}_k^{s}\mathbf{V}_{\text{RF}}\Big(\sum_{j=1}^{K}\hspace{-0.1cm}\mathbf{V}_j^s\mathbf{V}_j^{s^H}\hspace{-0.1cm}\Big)\mathbf{V}_{\text{RF}}^H\mathbf{H}_k^{{s}^H}\hspace{-0.2cm}+\hspace{-0.1cm}\sigma_{\text{noise},k}^{s\,2}\mathbf{I}_{N_r}\Big)\nonumber\\
	&\times\mathbf{U}_{\text{RF},k}\Big)^{-1}\mathbf{U}_{\text{RF},k}^H\mathbf{H}_k^{s}\mathbf{V}_{\text{RF}}\mathbf{V}_k^s.\label{eq:U}
\end{align}

\subsection{The Subproblem with Respect to $\mathbf{W}_k^s$}\label{sec:update_W}
Substituting the MMSE receiver update \eqref{eq:U} into \eqref{eq:E}, and using \(\sigma_{{\rm noise},k}^{s\,2}>0\) together with the full-column-rank assumption on \(\mathbf U_{{\rm RF},k}\), gives \(\mathbf E_k^s\succ\mathbf0\). Therefore, the minimization over \(\mathbf W_k^s\succ\mathbf0\) is well defined, and the first-order optimality condition yields the unique update \(\mathbf W_k^s=(\mathbf E_k^s)^{-1}\).

\subsection{A Three-Block ADMM Reformulation for $\{\mathbf V_k^s\}$}
\label{subsec:admm_blocks_hyb_provable}

For fixed $\mathbf V_{\mathrm{RF}}$, $\{\mathbf U_{\mathrm{RF},k}\}_{k=1}^K$,
$\{\mathbf U_k^s\}_{k=1,s\in\mathcal S}^{K}$, and
$\{\mathbf W_k^s\}_{k=1,s\in\mathcal S}^{K}$, the update of
$\{\mathbf V_k^s\}_{k=1,s\in\mathcal S}^{K}$ in \eqref{opt:main-revised} is convex.
Define $
\mathbf G_k^s \triangleq
\mathbf V_{\mathrm{RF}}^H\mathbf H_k^{sH}\mathbf U_{\mathrm{RF},k}\mathbf U_k^s
$ and $
\mathbf \Psi^s \triangleq \sum_{j=1}^K \mathbf G_j^s\mathbf W_j^s\mathbf G_j^{sH}.$
Then, up to an additive constant independent of $\{\mathbf V_k^s\}$, the WMMSE objective is
\begin{equation}
	\begin{split}
		f(\{\mathbf V_k^s\})
		= \sum_{s\in\mathcal S}\sum_{k=1}^K
		&\left(
		\operatorname{tr}\!\left(\mathbf V_k^{sH}\mathbf \Psi^s\mathbf V_k^s\right)
		\right.\\
		&\left.
		-2\Re\!\left\{\operatorname{tr}\!\left((\mathbf G_k^s\mathbf W_k^s)^H\mathbf V_k^s\right)\right\}
		\right).
	\end{split}\nonumber
\end{equation}

To put the inner problem in a form compatible with the convergence theory of direct
multi-block ADMM, we add a quadratic regularization with parameter $\eta_v>0$ and
introduce two auxiliary copies $\mathbf R_k^s$ and $\mathbf Z_k^s$. Thus, $
f_{\eta_v}(\{\mathbf V_k^s\})
\triangleq
f(\{\mathbf V_k^s\})+\frac{\eta_v}{2}\sum_{s\in\mathcal S}\sum_{k=1}^K\|\mathbf V_k^s\|_F^2.$
For brevity, let $\mathbf V$, $\mathbf R$, and $\mathbf Z$ denote the collections
$\{\mathbf V_k^s\}_{k,s}$, $\{\mathbf R_k^s\}_{k,s}$, and $\{\mathbf Z_k^s\}_{k,s}$,
respectively. As for partially-connected architecture, $\mathbf{V}_{\text{RF}}^H\mathbf{V}_{\text{RF}} = \frac{N_t}{N_{\text{RF}}}\mathbf{I}_{N_{\text{RF}}}$ \cite[Lemma 1]{7913599}, the power-feasible set is
\begin{equation}
	\mathcal P \triangleq
	\left\{
	\mathbf V:
	\frac{N_t}{N_{\mathrm{RF}}}\sum_{k=1}^K \operatorname{tr}\!\big(\mathbf V_k^{sH}\mathbf V_k^s\big)\le P^s,
	\ \forall s\in\mathcal S
	\right\}.
	\nonumber
\end{equation}
The spectral-mask feasible set is
\begin{equation}
	\mathcal M \triangleq
	\big\{
	\mathbf R:
	\sum_{s\in\mathcal S} |\mathbf A_n[j,s]|^2 \sum_{k=1}^K \|\mathbf R_k^s[m,:]\|_2^2
	\le \frac{\mathbf{r}_j}{\ln(G/\epsilon_2)},
	\ \forall m,\ \forall j
	\big\}.\nonumber
\end{equation}
The clipping-feasible set is $
\mathcal C \triangleq
\big\{
\mathbf Z:
\sum_{k=1}^K\sum_{s\in\mathcal S}\|\mathbf Z_k^s[m,:]\|_2^2
\le \frac{\chi^2\ell S}{\ln(\ell S/\epsilon_1)},
\ \forall m
\big\}.$
For any closed convex set $\mathcal S$, let $\delta_{\mathcal S}(\mathbf Y)\triangleq 0$ if $\mathbf Y\in\mathcal S$, and $\delta_{\mathcal S}(\mathbf Y)\triangleq +\infty$ otherwise. With these definitions, the regularized digital-precoder update is
\begin{equation}
	\begin{aligned}
		\min_{\mathbf Z,\mathbf R,\mathbf V}\quad
		&\theta_1(\mathbf Z)+\theta_2(\mathbf R)+\theta_3(\mathbf V)\\
		\text{s.t.}\quad
		&\mathbf Z_k^s-\mathbf R_k^s=\mathbf 0,\mathbf R_k^s-\mathbf V_k^s=\mathbf 0,\: \forall k,\ \forall s\in\mathcal S,\\
	\end{aligned}
	\label{opt:inner3block_hyb}
\end{equation}
where
\begin{equation}
	\theta_1(\mathbf Z)\triangleq \delta_{\mathcal C}(\mathbf Z),
	\theta_2(\mathbf R)\triangleq \delta_{\mathcal M}(\mathbf R),
	\theta_3(\mathbf V)\triangleq f_{\eta_v}(\{\mathbf V_k^s\})+\delta_{\mathcal P}(\mathbf V).
	\label{eq:theta123_hyb}
\end{equation}

After separating real and imaginary parts, \eqref{opt:inner3block_hyb} is a
three-block separable convex problem with linear equality constraints. The
theorem-facing block order is $(\mathbf Z,\mathbf R,\mathbf V)$, so that the third
block $\theta_3$ is strongly convex.

Let $\langle \mathbf A,\mathbf B\rangle\triangleq \Re\{\operatorname{tr}(\mathbf A^H\mathbf B)\}$.
Introducing dual variables $\boldsymbol\Lambda_{1,k}^s$ and $\boldsymbol\Lambda_{2,k}^s$
for $\mathbf Z_k^s-\mathbf R_k^s=\mathbf 0$ and $\mathbf R_k^s-\mathbf V_k^s=\mathbf 0$,
the augmented Lagrangian is
\begin{equation}
	\begin{aligned}
		&	\mathcal L(\mathbf Z,\mathbf R,\mathbf V,\boldsymbol\Lambda_1,\boldsymbol\Lambda_2)
		={}\theta_1(\mathbf Z)+\theta_2(\mathbf R)+\theta_3(\mathbf V)\\
		&+\sum_{s\in\mathcal S}\sum_{k=1}^K
		\Big(
		\langle \boldsymbol\Lambda_{1,k}^s,\mathbf Z_k^s-\mathbf R_k^s\rangle
		+\frac{\rho}{2}\|\mathbf Z_k^s-\mathbf R_k^s\|_F^2\\
		&
		+\langle \boldsymbol\Lambda_{2,k}^s,\mathbf R_k^s-\mathbf V_k^s\rangle
		+\frac{\rho}{2}\|\mathbf R_k^s-\mathbf V_k^s\|_F^2
		\Big).
	\end{aligned}\nonumber
\end{equation}

Accordingly, at iteration $\tau+1$, the ADMM updates are
\begin{equation}	\label{eq:admm_3block_updates_hyb}
	\begin{split}
		\mathbf Z^{\tau+1}
		&=
		\arg\min_{\mathbf Z}\,
		\mathcal L\!\left(
		\mathbf Z,\mathbf R^\tau,\mathbf V^\tau,\boldsymbol\Lambda_1^\tau,\boldsymbol\Lambda_2^\tau
		\right),\\
		\mathbf R^{\tau+1}
		&=
		\arg\min_{\mathbf R}\,
		\mathcal L\!\left(
		\mathbf Z^{\tau+1},\mathbf R,\mathbf V^\tau,\boldsymbol\Lambda_1^\tau,\boldsymbol\Lambda_2^\tau
		\right),\\
		\mathbf V^{\tau+1}
		&=
		\arg\min_{\mathbf V}\,
		\mathcal L\!\left(
		\mathbf Z^{\tau+1},\mathbf R^{\tau+1},\mathbf V,\boldsymbol\Lambda_1^\tau,\boldsymbol\Lambda_2^\tau
		\right),\\
		\boldsymbol\Lambda_{1,k}^{s,\tau+1}
		&=
		\boldsymbol\Lambda_{1,k}^{s,\tau}
		+\rho\left(\mathbf Z_k^{s,\tau+1}-\mathbf R_k^{s,\tau+1}\right), \forall k,\ \forall s\in\mathcal S,\\
		\boldsymbol\Lambda_{2,k}^{s,\tau+1}
		&=
		\boldsymbol\Lambda_{2,k}^{s,\tau}
		+\rho\left(\mathbf R_k^{s,\tau+1}-\mathbf V_k^{s,\tau+1}\right),
		\forall k,\ \forall s\in\mathcal S.
	\end{split}
\end{equation}

The three primal updates are detailed next.

\subsubsection{Update of $\mathbf Z$}
\label{subsec:z_solution_hyb_provable}

The $\mathbf Z$-subproblem is the Euclidean projection onto $\mathcal C$. Indeed, for fixed
$\mathbf R^\tau$ and $\boldsymbol\Lambda_1^\tau$, it is
\(\mathbf Z^{\tau+1}=\arg\min_{\mathbf Z\in\mathcal C}\sum_{s\in\mathcal S}\sum_{k=1}^K\|\mathbf Z_k^s-\mathbf R_k^{s,\tau}+\tfrac{1}{\rho}\boldsymbol\Lambda_{1,k}^{s,\tau}\|_F^2.\)
Hence, with $\tilde{\mathbf Z}_k^{s,\tau}\triangleq
	\mathbf R_k^{s,\tau}-\frac{1}{\rho}\boldsymbol\Lambda_{1,k}^{s,\tau}$,
the update is simply the projection of $\{\tilde{\mathbf Z}_k^{s,\tau}\}_{k,s}$ onto the row-wise
energy set $\mathcal C$. Since $\mathcal C$ is separable across the RF-chain index $m$, the projection is
a per-row scaling $	\mathbf Z_k^{s,\tau+1}[m,:]
	=
	\zeta_m^\tau\,\tilde{\mathbf Z}_k^{s,\tau}[m,:], \forall k,\ \forall s\in\mathcal S$,
where for all $m$, we have
\begin{equation}
	\zeta_m^\tau
	\triangleq
	\min\!\left\{
	1,\
	\frac{\chi\sqrt{\ell S/\ln(\ell S/\epsilon_1)}}
	{\sqrt{\sum_{\kappa=1}^{K}\sum_{s\in\mathcal S}
			\|\tilde{\mathbf Z}_\kappa^{s,\tau}[m,:]\|_2^2}}
	\right\}.\nonumber
\end{equation}

\subsubsection{Update of $\mathbf R$}
\label{subsec:r_solution_hyb_provable}

For fixed $\mathbf Z^{\tau+1}$, $\mathbf V^\tau$, $\boldsymbol\Lambda_1^\tau$, and
$\boldsymbol\Lambda_2^\tau$, the $\mathbf R$-subproblem is
	\begin{align}
		\mathbf R^{\tau+1}
		=
		\arg\min_{\mathbf R\in\mathcal M}&\ 
		\sum_{s\in\mathcal S}\sum_{k=1}^K
		\Big(
		\|\mathbf Z_k^{s,\tau+1}-\mathbf R_k^s+\boldsymbol\Lambda_{1,k}^{s,\tau}/\rho\|_F^2\nonumber\\
		&+\|\mathbf R_k^s-\mathbf V_k^{s,\tau}+\boldsymbol\Lambda_{2,k}^{s,\tau}/\rho\|_F^2
		\Big).\nonumber
	\end{align}
Completing the square shows that this is the Euclidean projection onto $\mathcal M$ of \(\tilde{\mathbf R}_k^{s,\tau}\triangleq\tfrac{1}{2}(\mathbf Z_k^{s,\tau+1}+\mathbf V_k^{s,\tau}+\tfrac{1}{\rho}\boldsymbol\Lambda_{1,k}^{s,\tau}-\tfrac{1}{\rho}\boldsymbol\Lambda_{2,k}^{s,\tau}),\ \forall k,\ \forall s.\). Then, $\mathbf R^{\tau+1}$ is updated from the equivalent problem
\begin{align}
	\mathbf R^{\tau+1}
	&=
	\arg\min_{\mathbf R}\quad
	\sum_{s\in\mathcal S}\sum_{k=1}^K\|\mathbf R_k^s-\tilde{\mathbf R}_k^{s,\tau}\|_F^2
	\label{opt:R_proj_hyb_provable}\\
	\text{s.t.}\quad
	&\sum_{s\in\mathcal S}|\mathbf A_n[j,s]|^2\sum_{k=1}^K\|\mathbf R_k^s[m,:]\|_2^2
	\le \frac{\mathbf{r}_j}{\ln(G/\epsilon_2)},
	 \forall m,\ \forall j.
	\nonumber
\end{align}

Because both the objective and the constraints are separable across the RF-chain index $m$,
\eqref{opt:R_proj_hyb_provable} decomposes into $N_{\mathrm{RF}}$ independent convex quadratically constrained quadratic program (QCQPs).
Fix one RF chain $m$ and associate multipliers $\mu_{m,j}\ge 0$ with the $G$ mask inequalities.
The corresponding KKT conditions yield
\begin{equation}
	\mathbf R_k^s[m,:]
	=
	\frac{\tilde{\mathbf R}_k^{s,\tau}[m,:]}
	{1+\sum_{j=1}^{G}\mu_{m,j}\,|\mathbf A_n[j,s]|^2},
	\qquad \forall k,\ \forall s\in\mathcal S.
	\label{eq:R_primal_hyb_provable}
\end{equation}
Thus, once the dual vector $\boldsymbol\mu_m\triangleq(\mu_{m,1},\ldots,\mu_{m,G})$ is known,
the primal row update is available in closed form.

Substituting \eqref{eq:R_primal_hyb_provable} into the Lagrangian yields the dual
function \(g_m(\boldsymbol\mu_m)\), whose maximization is equivalent, up to an
additive constant independent of \(\boldsymbol\mu_m\), to minimizing
\begin{equation}
	f_m(\boldsymbol\mu_m)
	\triangleq
	\sum_{s\in\mathcal S}
	\frac{\sum_{k=1}^{K}\|\tilde{\mathbf R}_k^{s,\tau}[m,:]\|_2^2}
	{1+\sum_{j=1}^{G}\mu_{m,j}|\mathbf A_n[j,s]|^2}
	+
	\sum_{j=1}^{G}\mu_{m,j}\frac{\mathbf{r}_j}{\ln(G/\epsilon_2)}.
	\label{eq:fm_hyb_provable}
\end{equation}
Therefore, we solve the equivalent dual optimization problem $
\min_{\boldsymbol\mu_m\succeq \mathbf 0} f_m(\boldsymbol\mu_m).$
We solve this by exact cyclic CD, where $\kappa$ indexes the sweeps
over $\boldsymbol\mu_m$ and one sweep updates all $G$ coordinates in the
order $j=1,\ldots,G$. Fix all dual coordinates except the $j^{\text{th}}$
one, and define
\begin{align}
	&c_{m,j}^{\kappa}[s]\triangleq \nonumber\\
	&
	1
	+\sum_{j'<j}\mu_{m,j'}^{\kappa+1}|\mathbf A_n[j',s]|^2
	+\sum_{j'>j}\mu_{m,j'}^{\kappa}|\mathbf A_n[j',s]|^2, s\in\mathcal S.\nonumber
\end{align}
Then, the $j^{\text{th}}$ coordinate update solves the one-dimensional convex problem
\begin{equation}
	\hspace{-.2cm}\mu_{m,j}^{\kappa+1}
	=
	\arg\min_{t\ge 0}\ 
	\phi_{m,j}^{\kappa}(t)
	\triangleq
	\sum_{s\in\mathcal S}
	\frac{\sum_{k=1}^{K}\|\tilde{\mathbf R}_k^{s,\tau}[m,:]\|_2^2}
	{c_{m,j}^{\kappa}[s]+t|\mathbf A_n[j,s]|^2}
	+\frac{t\:\mathbf{r}_j}{\ln(G/\epsilon_2)}.
	\label{opt:mu_coord_hyb_provable}
\end{equation}
Its derivative is
\begin{equation}
	(\phi_{m,j}^{\kappa})'(t)
	=
	\frac{\mathbf{r}_j}{\ln(G/\epsilon_2)}
	-
	\sum_{s\in\mathcal S}
	\frac{\left(\sum_{k=1}^{K}\|\tilde{\mathbf R}_k^{s,\tau}[m,:]\|_2^2\right)|\mathbf A_n[j,s]|^2}
	{\left(c_{m,j}^{\kappa}[s]+t|\mathbf A_n[j,s]|^2\right)^2}.
	\label{eq:phi_grad_hyb_provable}
\end{equation}
Under the assumption that there exists at least one $s\in\mathcal S$ such that
$\big(\sum_{k=1}^{K}\|\tilde{\mathbf R}_k^{s,\tau}[m,:]\|_2^2\big)|\mathbf A_n[j,s]|^2>0$,
the derivative in \eqref{eq:phi_grad_hyb_provable} is continuous and strictly increasing in $t\ge 0$.
Therefore, the exact coordinate update is
\begin{equation}
	\mu_{m,j}^{\kappa+1}
	=
	\begin{cases}
		0, & \text{if } (\phi_{m,j}^{\kappa})'(0)\ge 0,\\[1mm]
		t_{m,j}^{\kappa}, & \text{otherwise},
	\end{cases}
	\label{eq:mu_update_hyb_provable}
\end{equation}
where $t_{m,j}^{\kappa}>0$ is the unique root of $(\phi_{m,j}^{\kappa})'(t)=0$, i.e.,
\begin{equation}
	\sum_{s\in\mathcal S}
	\frac{\left(\sum_{k=1}^{K}\|\tilde{\mathbf R}_k^{s,\tau}[m,:]\|_2^2\right)|\mathbf A_n[j,s]|^2}
	{\left(c_{m,j}^{\kappa}[s]+t_{m,j}^{\kappa}|\mathbf A_n[j,s]|^2\right)^2}
	=
	\frac{\mathbf{r}_j}{\ln(G/\epsilon_2)}.
	\label{eq:mu_root_hyb_provable}
\end{equation}
Since the left-hand side of \eqref{eq:mu_root_hyb_provable} is strictly decreasing in $t$,
the scalar root $t_{m,j}^{\kappa}$ is computed exactly by bisection. After the cyclic
CD iterations converge for each $m$, the primal variables
$\mathbf R_k^{s,\tau+1}[m,:]$ are recovered from \eqref{eq:R_primal_hyb_provable}.

\subsubsection{Update of $\mathbf V$}
\label{subsec:v_solution_hyb_provable}

For fixed $\mathbf R^{\tau+1}$ and $\boldsymbol\Lambda_2^\tau$, the $\mathbf V$-subproblem is
	\begin{align}
		\mathbf V^{\tau+1}
		=
		\arg\min_{\mathbf V\in\mathcal P}\ 
		&f_{\eta_v}(\{\mathbf V_k^s\})
		+\sum_{s\in\mathcal S}\sum_{k=1}^K
		\Big(
		\langle \boldsymbol\Lambda_{2,k}^{s,\tau},\mathbf R_k^{s,\tau+1}-\mathbf V_k^s\rangle\nonumber \\
		&
		+\frac{\rho}{2}\|\mathbf R_k^{s,\tau+1}-\mathbf V_k^s\|_F^2
		\Big).
	\end{align}
Since the power constraint is separable across $s$, the update decouples over subcarriers. For a fixed
$s\in\mathcal S$, the corresponding problem is
\begin{align}
	\min_{\{\mathbf V_k^s\}_{k=1}^K}\quad
	&\sum_{k=1}^K\left(
	\operatorname{tr}\!\left(\mathbf V_k^{sH}\mathbf \Psi^s\mathbf V_k^s\right)
	-2\Re\!\left\{\operatorname{tr}\!\left((\mathbf G_k^s\mathbf W_k^s)^H\mathbf V_k^s\right)\right\}
	\right)\nonumber\\
	&\hspace{-.5cm}+\sum_{k=1}^K\left(
	\frac{\eta_v}{2}\|\mathbf V_k^s\|_F^2
	+\frac{\rho}{2}\|\mathbf R_k^{s,\tau+1}-\mathbf V_k^s\|_F^2
	-\langle \boldsymbol\Lambda_{2,k}^{s,\tau},\mathbf V_k^s\rangle
	\right)\nonumber\\
	\text{s.t.}\quad
	&\frac{N_t}{N_{\mathrm{RF}}}\sum_{k=1}^K \operatorname{tr}\!\left(\mathbf V_k^{sH}\mathbf V_k^s\right)\le P^s.\nonumber
\end{align}
Introducing the dual variable $\vartheta^s\ge 0$ for the power constraint, the KKT conditions yield
\(\mathbf V_k^{s,\tau+1}=\bigl(\mathbf \Psi^s+(\tfrac{\eta_v+\rho}{2}+\vartheta^s\tfrac{N_t}{N_{\mathrm{RF}}})\mathbf I_{N_{\mathrm{RF}}}\bigr)^{-1}\bigl(\mathbf G_k^s\mathbf W_k^s+\tfrac{\rho}{2}\mathbf R_k^{s,\tau+1}+\tfrac{1}{2}\boldsymbol\Lambda_{2,k}^{s,\tau}\bigr),\ \forall k,\)
together with the feasibility condition
$\frac{N_t}{N_{\mathrm{RF}}}\sum_{k=1}^K \operatorname{tr}\!\left(\mathbf V_k^{sH}\mathbf V_k^s\right)\le P^s$
and the complementary slackness condition
$\vartheta^s\!\left(\frac{N_t}{N_{\mathrm{RF}}}\sum_{k=1}^K \operatorname{tr}\!\left(\mathbf V_k^{sH}\mathbf V_k^s\right)-P^s\right)=0$.
Thus, $\vartheta^s$ is obtained by scalar bisection, with $\vartheta^s=0$ whenever the unconstrained
solution already satisfies the power bound.

\begin{proposition}\label{prop:R_dual_cd_convergence_hyb_provable}
	Fix $m\in\{1,\ldots,N_{\mathrm{RF}}\}$ and consider $\min_{\boldsymbol\mu_m\succeq \mathbf 0} f_m(\boldsymbol\mu_m)$. Then: 1) the dual problem is convex and admits at least one minimizer $\boldsymbol\mu_m^\star\succeq\mathbf 0$; 2) the cyclic CD iterates generated by \eqref{eq:mu_update_hyb_provable} are well defined, monotonically nonincreasing in $f_m$, and every accumulation point is a dual minimizer \cite[Prop.~2.7.1]{bertsekas1999nonlinear}; and 3) for any such dual minimizer, the primal point recovered from \eqref{eq:R_primal_hyb_provable} is the unique global minimizer of the RF-chain-$m$ projection subproblem in \eqref{opt:R_proj_hyb_provable}, since this subproblem is a Euclidean projection onto a nonempty closed convex set.
\end{proposition}

\begin{proof}
	The function $f_m$ in \eqref{eq:fm_hyb_provable} is proper, continuous, and convex on $\mathbb R_+^G$, and satisfies $f_m(\boldsymbol\mu_m)\ge \sum_{j=1}^{G}\mu_{m,j}\frac{\mathbf{r}_j}{\ln(G/\epsilon_2)}$, hence is coercive and attains a minimum, proving 1).
	
	For each coordinate $j$, the one-dimensional subproblem \eqref{opt:mu_coord_hyb_provable} is convex and coercive on $[0,\infty)$ whenever there exists $s$ with $B_s^\tau|\mathbf A_n[j,s]|^2>0$ (where $B_s^\tau\triangleq \sum_{k=1}^K\|\tilde{\mathbf R}_k^{s,\tau}[m,:]\|_2^2$); otherwise, the optimum is $\mu_{m,j}^{\kappa+1}=0$. Hence the update \eqref{eq:mu_update_hyb_provable} is well defined and the objective is monotonically nonincreasing. By coercivity of $f_m$, the iterates remain in a bounded sublevel set, and standard convergence of cyclic CD for differentiable convex functions with separable nonnegativity constraints \cite[Prop.~2.7.1]{bertsekas1999nonlinear} ensures that every accumulation point is a dual minimizer, proving 2).
	
	Finally, the RF-chain-$m$ primal problem in \eqref{opt:R_proj_hyb_provable} is a Euclidean projection onto a nonempty closed convex set, hence admits a unique minimizer. Slater's condition holds since $\mathbf 0$ is strictly feasible, so strong duality applies and the KKT stationarity condition gives \eqref{eq:R_primal_hyb_provable}; substituting any dual minimizer $\boldsymbol\mu_m^\star$ recovers this unique primal point, proving 3).
\end{proof}

\begin{proposition}\label{prop:admm_convergence_hyb_provable}
For any matrix \(\mathbf Y\), let \(\mathcal R(\mathbf Y)\triangleq[\Re\{\operatorname{vec}(\mathbf Y)\}^T,\Im\{\operatorname{vec}(\mathbf Y)\}^T]^T\), and for any collection \(\mathbf Y=\{\mathbf Y_k^s\}_{k,s}\), let \(\mathcal R(\mathbf Y)\) denote the stacking of \(\mathcal R(\mathbf Y_k^s)\) in a fixed order. Then \(\boldsymbol{\beta}_1\triangleq\mathcal R(\mathbf Z)\), \(\boldsymbol{\beta}_2\triangleq\mathcal R(\mathbf R)\), and \(\boldsymbol{\beta}_3\triangleq\mathcal R(\mathbf V)\).
If the inner CD loop in the $\mathbf R$-update is run to convergence, and if $\rho$ satisfies the admissible range of \cite[Eq.~(3.37)]{tao2018convergence} for the realified three-block problem associated with \eqref{opt:inner3block_hyb}, then the ADMM iterates in \eqref{eq:admm_3block_updates_hyb} converge to a primal--dual solution of \eqref{opt:inner3block_hyb}. Consequently, the primal limit $(\mathbf Z^\star,\mathbf R^\star,\mathbf V^\star)$ is a global optimal solution of \eqref{opt:inner3block_hyb} and satisfies its KKT conditions. Moreover, if $\{\eta_v^\iota\}_\iota$ is any sequence with $\eta_v^\iota>0$ and $\eta_v^\iota\to 0$, and $(\mathbf Z^\iota,\mathbf R^\iota,\mathbf V^\iota)$ denotes a global optimal solution of \eqref{opt:inner3block_hyb} with parameter $\eta_v^\iota$, then every accumulation point of $\{(\mathbf Z^\iota,\mathbf R^\iota,\mathbf V^\iota)\}_\iota$ is a global optimal solution of the unregularized three-block inner problem obtained by setting $\eta_v=0$.
\end{proposition}

\begin{proof}
	Because the extended map \(\mathcal R\) is a real-linear bijection, it suffices to work with the realified problem. By \eqref{eq:theta123_hyb}, $\theta_1=\delta_{\mathcal C}$ and $\theta_2=\delta_{\mathcal M}$ are closed proper convex, while $\theta_3$ is strongly convex with modulus at least $\eta_v$. The realified constraints are $[\mathbf I\ -\mathbf I\ \mathbf 0;\ \mathbf 0\ \mathbf I\ -\mathbf I][\boldsymbol{\beta}_1^T,\boldsymbol{\beta}_2^T,\boldsymbol{\beta}_3^T]^T=\mathbf 0$, and each block column has full column rank. Since the mask, clipping, and power budgets are strictly positive, the origin is feasible, and the feasible set of \eqref{opt:inner3block_hyb} is nonempty. The feasible set is also nonempty and closed: closedness follows from the closedness of $\mathcal C$, $\mathcal M$, and $\mathcal P$ together with the linear constraints. The constraint matrices $A_1=[\mathbf I;\mathbf 0]$, $A_2=[-\mathbf I;\mathbf I]$, and $A_3=[\mathbf 0;-\mathbf I]$ are all full column rank. Hence, the realified problem satisfies the assumptions of \cite[Thm.~3.1]{tao2018convergence}, with admissible $\rho$ given by \cite[Eq.~(3.37)]{tao2018convergence}; for the present three-block setting this reduces to $\rho\in(0,2\eta_v/5)$. Boundedness of the ADMM iterates is not assumed but follows directly from \cite[Eqs.~(3.46)--(3.48)]{tao2018convergence}: under strong convexity of $\theta_3$ with modulus $\eta_v$, full column rank of all constraint blocks, and $\rho$ in the admissible range, the Lyapunov function is non-increasing along iterates, which yields iterate boundedness.
	
	By Proposition~\ref{prop:R_dual_cd_convergence_hyb_provable}, if the inner CD loop is run to convergence, then each $\mathbf R$-subproblem is solved exactly. Therefore, \eqref{eq:admm_3block_updates_hyb} is an exact ADMM scheme for the regularized realified three-block convex problem, so \cite[Thm.~3.1]{tao2018convergence} yields convergence of the primal--dual iterates to a primal--dual solution. Since the problem is convex, the primal limit is globally optimal and satisfies the KKT conditions.
	
	For the final statement, let $z\triangleq(\mathbf Z,\mathbf R,\mathbf V)$, let $\mathcal F$ be the feasible set of the unregularized problem $(\eta_v=0)$, let $F_0$ denote its objective, and define $R^\iota(z)\triangleq \frac{\eta_v^\iota}{2}\sum_{s,k}\|\mathbf V_k^s\|_F^2$. Because $\mathcal F$ is compact and independent of $\eta_v$, any sequence of regularized minimizers $\{z^\iota\}$ admits an accumulation point $\bar z\in\mathcal F$. For every $z\in\mathcal F$, optimality of $z^\iota$ gives $F_0(z^\iota)+R^\iota(z^\iota)\le F_0(z)+R^\iota(z)$. Along any convergent subsequence $\iota'$, boundedness of $\{z^{\iota'}\}$ and $\eta_v^{\iota'}\to 0$ imply $R^{\iota'}(z^{\iota'})\to 0$ and $R^{\iota'}(z)\to 0$ for every fixed $z\in\mathcal F$. Hence $F_0(\bar z)\le \liminf_{\iota'\to\infty}F_0(z^{\iota'})\le \limsup_{\iota'\to\infty}F_0(z^{\iota'})\le F_0(z)$ for all $z\in\mathcal F$, where the first inequality uses lower semicontinuity of $F_0$. Therefore, every accumulation point of regularized minimizers is globally optimal for the unregularized inner problem.
\end{proof}

	\vspace{-0.5cm}

	\subsection{The Subproblem with Respect to $\mathbf{V}_{\text{RF}}$}\label{sec:analogp}
	
Since \(\sum_{j=1}^K\mathbf V_j^s\mathbf V_j^{s^H}\) is Hermitian PSD, let \(\boldsymbol{\Theta}^s\) denote its Hermitian square root, so that \(\boldsymbol{\Theta}^s\boldsymbol{\Theta}^{s^H}=\sum_{j=1}^K\mathbf V_j^s\mathbf V_j^{s^H}\). Let \(\mathbf v_{\mathrm{PS}}\) collect the unit-modulus applied phase shifts to antennas and let \(m_i\) denote the RF-chain connected to antenna \(i\). Then, \(\mathbf V_{\mathrm{RF}}\boldsymbol{\Theta}^s\boldsymbol{\Theta}^{s^H}\mathbf V_{\mathrm{RF}}^H=\operatorname{diag}(\mathbf v_{\mathrm{PS}})\boldsymbol{\Gamma}^s\operatorname{diag}(\mathbf v_{\mathrm{PS}})^H\), where \(\boldsymbol{\Gamma}^s[i,j]\triangleq \boldsymbol{\Theta}^s[m_i,:](\boldsymbol{\Theta}^s[m_j,:])^H=(\boldsymbol{\Theta}^s\boldsymbol{\Theta}^{s^H})[m_i,m_j]\).
	
	We simplify the expected MSE and keep those terms that contain $\mathbf{V}_{\text{RF}}$ as follows:
	\begin{align}
		&\mathbf{E}_{k}^{s, \text{simplified}} =-\mathbf{U}_k^{s^H}\mathbf{U}_{\text{RF},k}^{H} \mathbf{H}_k^s\mathbf{V}_{\text{RF}} \mathbf{V}_k^s\nonumber\\
		& -(\mathbf{U}_k^{s^H} \mathbf{U}_{\text{RF},k}^{H}\mathbf{H}_k^s\mathbf{V}_{\text{RF}}\mathbf{V}_k^s)^H \nonumber\\
		&+ \mathbf{U}_k^{s^H}\mathbf{U}_{\text{RF},k}^{H} \Big( \sum_{j=1}^K \mathbf{H}_k^s\mathbf{V}_{\text{RF}}\mathbf{V}_j^s \mathbf{V}_j^{s^H} \mathbf{V}_{\text{RF}}^H \mathbf{H}_k^{s^H}  \Big)\mathbf{U}_{\text{RF},k} \mathbf{U}_k^s\nonumber\\
		&=-\mathbf{U}_k^{s^H}\mathbf{U}_{\text{RF},k}^{H} \mathbf{H}_k^s\mathbf{V}_{\text{RF}} \mathbf{V}_k^s -(\mathbf{U}_k^{s^H} \mathbf{U}_{\text{RF},k}^{H}\mathbf{H}_k^s\mathbf{V}_{\text{RF}}\mathbf{V}_k^s)^H \nonumber\\
		&+ \mathbf{U}_k^{s^H}\mathbf{U}_{\text{RF},k}^{H} \Big(  \mathbf{H}_k^s\text{diag}(\mathbf{v}_{\text{PS}}) \boldsymbol{\Gamma}^s \text{diag}(\mathbf{v}_{\text{PS}}) ^H\mathbf{H}_k^{s^H}  \Big)\mathbf{U}_{\text{RF},k} \mathbf{U}_k^s.\nonumber
	\end{align}
	We simplify the linear term in the above expression as $\sum_{s\in \mathcal{S}} \sum_{k=1}^{K}\text{tr}(\mathbf{V}_k^s\mathbf{W}_k^s\mathbf{U}_k^{s^H}\mathbf{U}_{\text{RF},k}^{H} \mathbf{H}_k^s\mathbf{V}_{\text{RF}}  ) = \bar{\mathbf{u}}_{\text{PS}}^H\bar{\mathbf{v}}_{\text{PS}}$, where $\bar{\mathbf{u}}_{\text{PS}} \triangleq \sum_{s\in \mathcal{S}} \sum_{k=1}^{K}\text{vec}((\mathbf{V}_k^s\mathbf{W}_k^s\mathbf{U}_k^{s^H}\mathbf{U}_{\text{RF},k}^{H} \mathbf{H}_k^s)^H)$ and $\bar{\mathbf{v}}_{\text{PS}}\triangleq\text{vec}(\mathbf{V}_{\text{RF}})$. Let \texttt{idx} index the nonzero entries of $\text{vec}(\mathbf{V}_{\text{RF}})$, with $|\texttt{idx}|=N_t$ under the partially-connected transmitter. Then $\mathbf{u}_{\text{PS}}=\bar{\mathbf{u}}_{\text{PS}}[\texttt{idx}]$. Therefore, we can rewrite the above equation as $\mathbf{u}_{\text{PS}}^{H}\mathbf{v}_{\text{PS}}$.
	
	Next, we simplify the quadratic term of $\mathbf{V}_{\text{RF}}$. For a diagonal matrix $\boldsymbol{\Pi}=\text{diag}(\boldsymbol{\pi})$, we have a useful equality $\text{tr}(\boldsymbol{\Pi}^H\mathbf{A}\boldsymbol{\Pi} \mathbf{B})=\boldsymbol{\pi}^H(\mathbf{A}\odot\mathbf{B}^T)\boldsymbol{\pi}$, where $\odot$ denotes the Hadamard product. We use this equality to apply the following simplification of the quadratic term in $\mathbf{v}_{\text{PS}}$:
	\begin{align}
		& \hspace{-0.5cm} \sum_{s\in \mathcal{S}} \sum_{k=1}^{K}\text{tr}\left( \text{diag}(\mathbf{v}_{\text{PS}}^{H})(\mathbf{H}_k^{s^H}  \mathbf{U}_{\text{RF},k} \mathbf{U}_k^s\mathbf{W}_k^s\mathbf{U}_k^{s^H}\mathbf{U}_{\text{RF},k}^{H}   \mathbf{H}_k^s)\text{diag}(\mathbf{v}_{\text{PS}}) \boldsymbol{\Gamma}^s\right)\nonumber\\
		&=\mathbf{v}_{\text{PS}}^{H}(\underbrace{\sum_{s\in \mathcal{S}} \sum_{k=1}^{K}((\mathbf{H}_k^{s^H}  \mathbf{U}_{\text{RF},k} \mathbf{U}_k^s\mathbf{W}_k^s\mathbf{U}_k^{s^H}\mathbf{U}_{\text{RF},k}^{H}   \mathbf{H}_k^s)\odot \boldsymbol{\Gamma}^{s^T})}_{\mathbf{Q}_{\text{PS}}})\mathbf{v}_{\text{PS}}.\nonumber
	\end{align}
	Using the above equations, those terms of the simplified weighted MSE that include $\mathbf{v}_\text{PS}$ can be rewritten as follows:
	\begin{align}
		&f(\mathbf{v}_{\text{PS}})=\mathbf{v}_{\text{PS}}^{H}\mathbf{Q}_{\text{PS}}\mathbf{v}_{\text{PS}} -2\Re\Big(\mathbf{u}_{\text{PS}}^{H}\mathbf{v}_{\text{PS}}\Big).\label{opt:main-revisedps}
	\end{align}
	Due to the nonconvex unit-magnitude constraints, we deploy a CD, updating one element $\mathbf{v}_{\text{PS}}[a]$ at a time while fixing all others. Expanding \eqref{opt:main-revisedps} and isolating the terms that depend on $\mathbf{v}_{\text{PS}}[a]$, we use the Hermitian property of $\mathbf{Q}_{\text{PS}}$ to write
	\begin{align}
		&f(\mathbf{v}_{\text{PS}})
		=
		\text{const}
		+\mathbf{Q}_{\text{PS}}[a,a]\nonumber\\
&		+2\Re\!\left(
		\left(
		\mathbf{Q}_{\text{PS}}[a,:]\mathbf{v}_{\text{PS}}
		-\mathbf{Q}_{\text{PS}}[a,a]\mathbf{v}_{\text{PS}}[a]
		-\mathbf{u}_{\text{PS}}[a]
		\right)^{*}\mathbf{v}_{\text{PS}}[a]
		\right).\nonumber
	\end{align}
	Since $|\mathbf{v}_{\text{PS}}[a]|=1$, the minimizer aligns the phase of $\mathbf{v}_{\text{PS}}[a]$ opposite to the linear coefficient, giving
	\begin{equation}
		\mathbf{v}_{\text{PS}}[a]
		=
		-\frac{
			\mathbf{Q}_{\text{PS}}[a,:]\mathbf{v}_{\text{PS}}
			-\mathbf{Q}_{\text{PS}}[a,a]\mathbf{v}_{\text{PS}}[a]
			-\mathbf{u}_{\text{PS}}[a]
		}{
			\left|
			\mathbf{Q}_{\text{PS}}[a,:]\mathbf{v}_{\text{PS}}
			-\mathbf{Q}_{\text{PS}}[a,a]\mathbf{v}_{\text{PS}}[a]
			-\mathbf{u}_{\text{PS}}[a]
			\right|
		}. \nonumber
	\end{equation}
We use CD, updating each $\mathbf{v}_{\text{PS}}[a]$ in turn (using current values for indices $<a$ and previous values for $>a$); if the numerator vanishes, $\mathbf{v}_{\text{PS}}[a]$ is left unchanged. Each update preserves $|\mathbf{v}_{\text{PS}}[a]|=1$ and is non-increasing in $f(\mathbf{v}_{\text{PS}})$, which is bounded below; hence the objective sequence converges. The iterates lie on the compact torus $\{|\mathbf{v}_{\text{PS}}[a]|=1\}^{N_t}$, so accumulation points exist. By smoothness of $f$ and separability of the unit-modulus constraints, every accumulation point is a coordinatewise minimizer and therefore, by \cite[Lem.~3.1]{tseng2001convergence}, a stationary point of the RF subproblem.
	
	PS impairments introduce independent random phase errors at each antenna element, modeled as complex exponentials with Gaussian-distributed phase deviations. Unlike common local oscillator (LO) phase noise affecting all antennas equally, PS errors are element-specific and occur when new precoding matrices are applied, destroying the intended precoding pattern through uncorrelated and random phase deviations \cite{9767793,icc2026}. Starting from the given objective function \eqref{opt:main-revisedps}, taking the expectation with respect to the phase errors in $\mathbf{v}_{\text{PS}}$ results in
	\[
	\mathbb{E}_{\mathbf{e}_{\text{error}}}[f(\mathbf{v}_{\text{PS}})] = \mathbb{E}_{\mathbf{e}_{\text{error}}}\left[\mathbf{v}_{\text{PS}}^{H}\mathbf{Q}_{\text{PS}}\mathbf{v}_{\text{PS}}\right] - 2\mathbb{E}_{\mathbf{e}_{\text{error}}}\left[\Re\Big(\mathbf{u}_{\text{PS}}^{H}\mathbf{v}_{\text{PS}}\Big)\right].
	\]
	Let $\mathbf{v}_{\text{PS}} = \tilde{\mathbf{v}}_{\text{PS}} \odot \mathbf{e}_{\text{error}}$, where $\tilde{\mathbf{v}}_{\text{PS}}$ denotes the vector of intended phase shifts and $\mathbf{e}_{\text{error}} = [e^{\jmath\Delta\theta_1}, e^{\jmath\Delta\theta_2}, \ldots, e^{\jmath\Delta\theta_{N_t}}]^T$ with $\Delta\theta_i \sim \mathcal{N}(0, \sigma_e^2)$.
	For the first term, we obtain
	\begin{align}
		\mathbb{E}_{\mathbf{e}_{\text{error}}}&\left[\mathbf{v}_{\text{PS}}^{H}\mathbf{Q}_{\text{PS}}\mathbf{v}_{\text{PS}}\right] = \mathbb{E}_{\mathbf{e}_{\text{error}}}\left[(\tilde{\mathbf{v}}_{\text{PS}} \odot \mathbf{e}_{\text{error}})^{H}\mathbf{Q}_{\text{PS}}(\tilde{\mathbf{v}}_{\text{PS}} \odot \mathbf{e}_{\text{error}})\right]\nonumber \\
		&= \sum_{i=1,j=1}^{N_t} \mathbf{Q}_{\text{PS}}[i,j] \tilde{\mathbf{v}}_{\text{PS}}^*[i] \tilde{\mathbf{v}}_{\text{PS}}[j] \mathbb{E}_{\mathbf{e}_{\text{error}}}[e^{-\jmath\Delta\theta_i}e^{\jmath\Delta\theta_j}]\nonumber
	\end{align}
	Since $\Delta\theta_i \sim \mathcal{N}(0, \sigma_e^2)$ are independent, using the moment-generating function, we have $\mathbb{E}_{\mathbf{e}_{\text{error}}}[e^{\jmath\Delta\theta_j}]=e^{-\sigma_e^2/2}$ and $\mathbb{E}_{\mathbf{e}_{\text{error}}}[e^{\jmath(\Delta\theta_j-\Delta\theta_i)}] = 1$ if $i = j$, and $\mathbb{E}_{\mathbf{e}_{\text{error}}}[e^{\jmath(\Delta\theta_j-\Delta\theta_i)}] = \mathbb{E}_{\mathbf{e}_{\text{error}}}[e^{\jmath\Delta\theta_j}]\mathbb{E}_{\mathbf{e}_{\text{error}}}[e^{-\jmath\Delta\theta_i}] = e^{-\sigma_e^2}$ if $i \neq j$.
	Therefore, we obtain
	\begin{align}
		&\mathbb{E}_{\mathbf{e}_{\text{error}}}\left[\mathbf{v}_{\text{PS}}^{H}\mathbf{Q}_{\text{PS}}\mathbf{v}_{\text{PS}}\right] = \sum_{i=1}^{N_t} \mathbf{Q}_{\text{PS}}[i,i] |\tilde{\mathbf{v}}_{\text{PS}}[i]|^2 \nonumber \\
		&+ e^{-\sigma_e^2}\sum_{i=1,j=1,i \neq j}^{N_t} \mathbf{Q}_{\text{PS}}[i,j] \tilde{\mathbf{v}}_{\text{PS}}^*[i] \tilde{\mathbf{v}}_{\text{PS}}[j]= \tilde{\mathbf{v}}_{\text{PS}}^{H}(\mathbf{Q}_{\text{PS}} \odot \mathbf{I}_{N_t})\tilde{\mathbf{v}}_{\text{PS}} \nonumber\\
		&+ e^{-\sigma_e^2}\tilde{\mathbf{v}}_{\text{PS}}^{H}(\mathbf{Q}_{\text{PS}} \odot (\mathbf{1}_{N_t } - \mathbf{I}_{N_t}))\tilde{\mathbf{v}}_{\text{PS}},\nonumber
	\end{align}
	where $\mathbf{1}_{N_t }$ is the all-ones matrix. The second term is
$
	\mathbb{E}_{\mathbf{e}_{\text{error}}}\left[\Re\Big(\mathbf{u}_{\text{PS}}^{H}\mathbf{v}_{\text{PS}}\Big)\right]
	=
	e^{-\sigma_e^2/2}\Re\Big(\mathbf{u}_{\text{PS}}^{H}\tilde{\mathbf{v}}_{\text{PS}}\Big).$
	The final simplified expectation becomes
	\begin{align}
		&	\mathbb{E}_{\mathbf{e}_{\text{error}}}[f(\mathbf{v}_{\text{PS}})] =  - 2e^{-\sigma_e^2/2}\Re\Big(\mathbf{u}_{\text{PS}}^{H}\tilde{\mathbf{v}}_{\text{PS}}\Big)\nonumber \\
		&\quad +\tilde{\mathbf{v}}_{\text{PS}}^{H}\underbrace{\left[\mathbf{Q}_{\text{PS}} \odot \mathbf{I}_{N_t} + e^{-\sigma_e^2}(\mathbf{Q}_{\text{PS}} \odot (\mathbf{1}_{N_t} - \mathbf{I}_{N_t}))\right]}_{\tilde{\mathbf{Q}}_{\text{PS}}}\tilde{\mathbf{v}}_{\text{PS}}.\nonumber
	\end{align}
	Similar to the analysis given before, the update rule for each phase shift can be obtained as follows:
	\begin{equation}
		\tilde{\mathbf{v}}_{\text{PS}}[a]
		=
		-\frac{
			\tilde{\mathbf{Q}}_{\text{PS}}[a,:]\tilde{\mathbf{v}}_{\text{PS}}
			-\tilde{\mathbf{Q}}_{\text{PS}}[a,a]\tilde{\mathbf{v}}_{\text{PS}}[a]
			-e^{-\sigma_e^2/2} \mathbf{u}_{\text{PS}}[a]
		}{
			\left|
			\tilde{\mathbf{Q}}_{\text{PS}}[a,:]\tilde{\mathbf{v}}_{\text{PS}}
			-\tilde{\mathbf{Q}}_{\text{PS}}[a,a]\tilde{\mathbf{v}}_{\text{PS}}[a]
			-e^{-\sigma_e^2/2} \mathbf{u}_{\text{PS}}[a]
			\right|
		}. \nonumber
	\end{equation}
	
	Although it was assumed that the phase error of each PS follows a Gaussian distribution, one can repeat the above analysis for any desired distribution.
	
\subsection{The Subproblem with Respect to RF Combining $\mathbf{U}_{\text{RF},k}$}\label{sec:updateU}

To optimize the RF combiner at user $k$, we need to minimize the following expression with respect to $\mathbf{U}_{\text{RF},k}$:
\begin{align}
	&-\sum_{s\in \mathcal{S}}\text{tr}(\mathbf{W}_k^s\mathbf{U}_k^{s^H}\mathbf{U}_{\text{RF},k}^{H} \mathbf{H}_k^s\mathbf{V}_{\text{RF}} \mathbf{V}_k^s)\label{eq:simpcost}\\
	&-\sum_{s\in \mathcal{S}}\text{tr}\!\big(\mathbf{W}_k^s(\mathbf{U}_k^{s^H} \mathbf{U}_{\text{RF},k}^{H}\mathbf{H}_k^s\mathbf{V}_{\text{RF}}\mathbf{V}_k^s)^H\big)\nonumber\\
	&+\sum_{s\in \mathcal{S}}\text{tr}\!\Big(\mathbf{W}_k^s\mathbf{U}_k^{s^H}\mathbf{U}_{\text{RF},k}^{H}
	\mathbf{H}_k^s\mathbf{V}_{\text{RF}}\Big( \sum_{j=1}^K \mathbf{V}_j^s \mathbf{V}_j^{s^H} \Big)\mathbf{V}_{\text{RF}}^H \mathbf{H}_k^{s^H} \mathbf{U}_{\text{RF},k} \mathbf{U}_k^s\Big)\nonumber\\
	&+\sum_{s\in\mathcal S}\sigma_{\text{noise},k}^{s\,2}\,
	\text{tr}\!\Big(
	\mathbf W_k^s\mathbf U_k^{s^H}\mathbf U_{\text{RF},k}^H
	\mathbf U_{\text{RF},k}\mathbf U_k^s
	\Big).\nonumber
\end{align}
We use CD and optimize one entry at a time while fixing the remaining entries of $\mathbf{U}_{\text{RF},k}$. Let $\bar{\mathbf{u}}_{\text{RF},k}\triangleq \text{vec}(\mathbf{U}_{\text{RF},k})$, and let \texttt{idx}$_k$ denote the indices of the feasible non-zero entries of $\bar{\mathbf{u}}_{\text{RF},k}$ induced by $\mathcal{E}_k$. Define $\mathbf{u}_{\text{RF},k}\triangleq \bar{\mathbf{u}}_{\text{RF},k}[\texttt{idx}_k]$, where $|\mathbf{u}_{\text{RF},k}[h]|=1$ for all $h$.

We first simplify the linear term. Using the cyclicity of the trace, $
\sum_{s\in \mathcal{S}}\text{tr}(\mathbf{W}_k^s\mathbf{U}_k^{s^H}\mathbf{U}_{\text{RF},k}^{H} \mathbf{H}_k^s\mathbf{V}_{\text{RF}} \mathbf{V}_k^s)
=
\sum_{s\in \mathcal{S}}\text{tr}(\mathbf{U}_{\text{RF},k}^{H} \mathbf{H}_k^s\mathbf{V}_{\text{RF}}\mathbf{V}_k^s\mathbf{W}_k^s\mathbf{U}_k^{s^H}).$
Define $\mathbf{B}_{\text{RF},k}\triangleq \sum_{s\in \mathcal{S}} \mathbf{H}_k^s\mathbf{V}_{\text{RF}}\mathbf{V}_k^s\mathbf{W}_k^s\mathbf{U}_k^{s^H}$. Then,
\begin{align}
&\sum_{s\in \mathcal{S}}\text{tr}(\mathbf{W}_k^s\mathbf{U}_k^{s^H}\mathbf{U}_{\text{RF},k}^{H} \mathbf{H}_k^s\mathbf{V}_{\text{RF}} \mathbf{V}_k^s)
=
\text{tr}(\mathbf{U}_{\text{RF},k}^{H}\mathbf{B}_{\text{RF},k})\nonumber\\
&=
\bar{\mathbf{u}}_{\text{RF},k}^H\text{vec}(\mathbf{B}_{\text{RF},k}).
\end{align}
Let $\bar{\mathbf{d}}_{\text{RF},k}\triangleq \text{vec}(\mathbf{B}_{\text{RF},k})$ and $\mathbf{d}_{\text{RF},k}\triangleq \bar{\mathbf{d}}_{\text{RF},k}[\texttt{idx}_k]$. Therefore, the contribution to the objective function from the linear terms is \( -2\Re(\mathbf{d}_{\text{RF},k}^{H}\mathbf{u}_{\text{RF},k}) \).

We simplify the term of \eqref{eq:simpcost} that is quadratic in $\mathbf{U}_{\text{RF},k}$ by defining $
\mathbf{O}_k^s\triangleq \mathbf{H}_k^s\mathbf{V}_{\text{RF}}\Big( \sum_{j=1}^K \mathbf{V}_j^s \mathbf{V}_j^{s^H} \Big)\mathbf{V}_{\text{RF}}^H \mathbf{H}_k^{s^H}+\sigma_{\text{noise},k}^{s\,2}\mathbf{I}_{N_r}.$
Using the identity $\text{tr}(\mathbf{V}\mathbf{B}\mathbf{V}^H\mathbf{A})=\text{vec}(\mathbf{V})^H(\mathbf{B}^T\otimes \mathbf{A})\text{vec}(\mathbf{V})$, we obtain
\[
\sum_{s\in \mathcal{S}}\text{tr}\!\Big(\mathbf{W}_k^s\mathbf{U}_k^{s^H}\mathbf{U}_{\text{RF},k}^{H} \mathbf{O}_k^s\mathbf{U}_{\text{RF},k}\mathbf{U}_k^s\Big)
=
\bar{\mathbf{u}}_{\text{RF},k}^{H}\mathbf{M}_{\text{RF},k}\bar{\mathbf{u}}_{\text{RF},k},
\]
where $
\mathbf{M}_{\text{RF},k}\triangleq \sum_{s\in \mathcal{S}}
\Big((\mathbf{U}_k^s\mathbf{W}_k^s\mathbf{U}_k^{s^H})^T\otimes \mathbf{O}_k^s\Big).$ Restricting to the feasible entries gives $\mathbf{Q}_{\text{RF},k}\triangleq \mathbf{M}_{\text{RF},k}[\texttt{idx}_k,\texttt{idx}_k]$, and hence the quadratic term can be written as $\mathbf{u}_{\text{RF},k}^{H}\mathbf{Q}_{\text{RF},k}\mathbf{u}_{\text{RF},k}$.

Using the above equations, those terms of the simplified weighted MSE that include $\mathbf{u}_{\text{RF},k}$ can be rewritten as $
f(\mathbf{u}_{\text{RF},k})
=
\mathbf{u}_{\text{RF},k}^{H}\mathbf{Q}_{\text{RF},k}\mathbf{u}_{\text{RF},k}
-2\Re(\mathbf{d}_{\text{RF},k}^{H}\mathbf{u}_{\text{RF},k}).$
To optimize with respect to the $(a,m)$-th feasible entry $\mathbf{U}_{\text{RF},k}[a,m]$, let $h$ denote its position in $\mathbf{u}_{\text{RF},k}$. Expanding the above objective and isolating the terms that depend on $\mathbf{u}_{\text{RF},k}[h]$, we use the Hermitian property of $\mathbf{Q}_{\text{RF},k}$ to write
\begin{align}
&f(\mathbf{u}_{\text{RF},k})
=
\mathrm{const}
+\mathbf{Q}_{\text{RF},k}[h,h]  +2\Re\!\Big(
\big(
\mathbf{Q}_{\text{RF},k}[h,:]\mathbf{u}_{\text{RF},k}      \nonumber\\
&
-\mathbf{Q}_{\text{RF},k}[h,h]\mathbf{u}_{\text{RF},k}[h] 
-\mathbf{d}_{\text{RF},k}[h]
\big)^*
\mathbf{u}_{\text{RF},k}[h]
\Big).\nonumber
\end{align}
Since $|\mathbf{u}_{\text{RF},k}[h]|=1$, the minimizer aligns the phase of $\mathbf{u}_{\text{RF},k}[h]$ opposite to the linear coefficient, giving
\begin{align}
	&\mathbf{U}_{\text{RF},k}[a,m]
	=
	\mathbf{u}_{\text{RF},k}[h]\nonumber\\
	&=
	-\frac{
		\mathbf{Q}_{\text{RF},k}[h,:]\mathbf{u}_{\text{RF},k}
		-\mathbf{Q}_{\text{RF},k}[h,h]\mathbf{u}_{\text{RF},k}[h]
		-\mathbf{d}_{\text{RF},k}[h]
	}{
		\left|
		\mathbf{Q}_{\text{RF},k}[h,:]\mathbf{u}_{\text{RF},k}
		-\mathbf{Q}_{\text{RF},k}[h,h]\mathbf{u}_{\text{RF},k}[h]
		-\mathbf{d}_{\text{RF},k}[h]
		\right|
	}.\nonumber
\end{align}
Sweeping CD over the feasible entries of $\mathbf{U}_{\text{RF},k}$ (leaving an entry unchanged when its numerator vanishes) preserves unit modulus and is non-increasing. The objective sequence is bounded below and converges; iterates lie on a compact set, so accumulation points exist. By smoothness and separability of the unit-modulus constraints, every accumulation point is a coordinatewise minimizer, hence stationary by \cite[Lem.~3.1]{tseng2001convergence}.

In the presence of the PS impairments at the user equipment that introduce random phase errors, one can use the robust optimization technique proposed for the RF precoder at the transmitter. In the vectorized objective above, taking the expectation with respect to the phase errors in $\mathbf{u}_{\text{RF},k}$ yields
\begin{align}
\mathbb{E}_{\mathbf{e}_{\text{error}}}&[f(\mathbf{u}_{\text{RF},k})]
=
\mathbb{E}_{\mathbf{e}_{\text{error}}}\!\left[\mathbf{u}_{\text{RF},k}^{H}\mathbf{Q}_{\text{RF},k}\mathbf{u}_{\text{RF},k}\right]
\nonumber\\
&-2\mathbb{E}_{\mathbf{e}_{\text{error}}}\!\left[\Re(\mathbf{d}_{\text{RF},k}^{H}\mathbf{u}_{\text{RF},k})\right].\nonumber
\end{align}
Let $\mathbf{u}_{\text{RF},k}= \tilde{\mathbf{u}}_{\text{RF},k}\odot \mathbf{e}_{\text{error}}$, where $\tilde{\mathbf{u}}_{\text{RF},k}$ denotes the vector of intended combining phase shifts and the entries of $\mathbf{e}_{\text{error}}$ are independent Gaussian random phase errors with variance $\sigma_e^2$. Then, exactly as in the RF-precoder case,
\[
\mathbb{E}_{\mathbf{e}_{\text{error}}}[f(\mathbf{u}_{\text{RF},k})]
=
-2e^{-\sigma_e^2/2}\Re(\mathbf{d}_{\text{RF},k}^{H}\tilde{\mathbf{u}}_{\text{RF},k})
+
\tilde{\mathbf{u}}_{\text{RF},k}^{H}\tilde{\mathbf{Q}}_{\text{RF},k}\tilde{\mathbf{u}}_{\text{RF},k},
\]
with $\tilde{\mathbf{Q}}_{\text{RF},k}\triangleq \mathbf{Q}_{\text{RF},k}\odot \mathbf{I}_{|\texttt{idx}_k|} + e^{-\sigma_e^2}(\mathbf{Q}_{\text{RF},k}\odot (\mathbf{1}_{|\texttt{idx}_k|}-\mathbf{I}_{|\texttt{idx}_k|}))$. Therefore, the robust update rule is
\begin{align}
&\tilde{\mathbf{U}}_{\text{RF},k}[a,m]
=
\tilde{\mathbf{u}}_{\text{RF},k}[h]
\nonumber\\
&=
-\frac{
	\tilde{\mathbf{Q}}_{\text{RF},k}[h,:]\tilde{\mathbf{u}}_{\text{RF},k}
	-\tilde{\mathbf{Q}}_{\text{RF},k}[h,h]\tilde{\mathbf{u}}_{\text{RF},k}[h]
	-e^{-\sigma_e^2/2}\mathbf{d}_{\text{RF},k}[h]
}{
	\left|
	\tilde{\mathbf{Q}}_{\text{RF},k}[h,:]\tilde{\mathbf{u}}_{\text{RF},k}
	-\tilde{\mathbf{Q}}_{\text{RF},k}[h,h]\tilde{\mathbf{u}}_{\text{RF},k}[h]
	-e^{-\sigma_e^2/2}\mathbf{d}_{\text{RF},k}[h]
	\right|
}.\nonumber
\end{align}

		\begin{figure*}[t!]
		\centering
		\begin{minipage}[t]{0.315\textwidth}
			\centering
			\includegraphics[width=\textwidth]{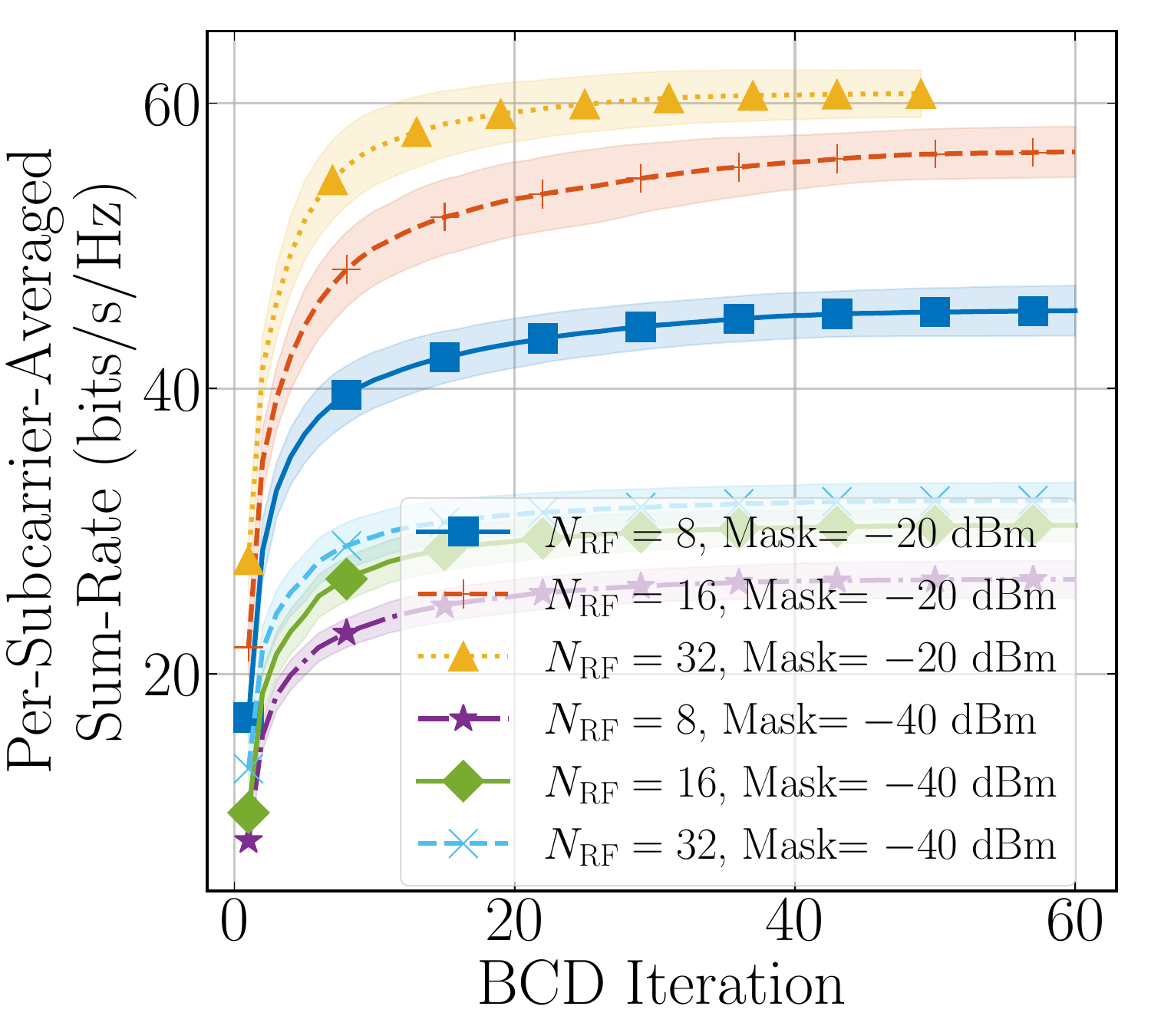}
			\captionof{figure}{The convergence of the proposed BCD approach.}
			\label{fig:conv_bcd}
		\end{minipage}
		\hfill
		\begin{minipage}[t]{0.315\textwidth}
			\centering
			\includegraphics[width=\textwidth]{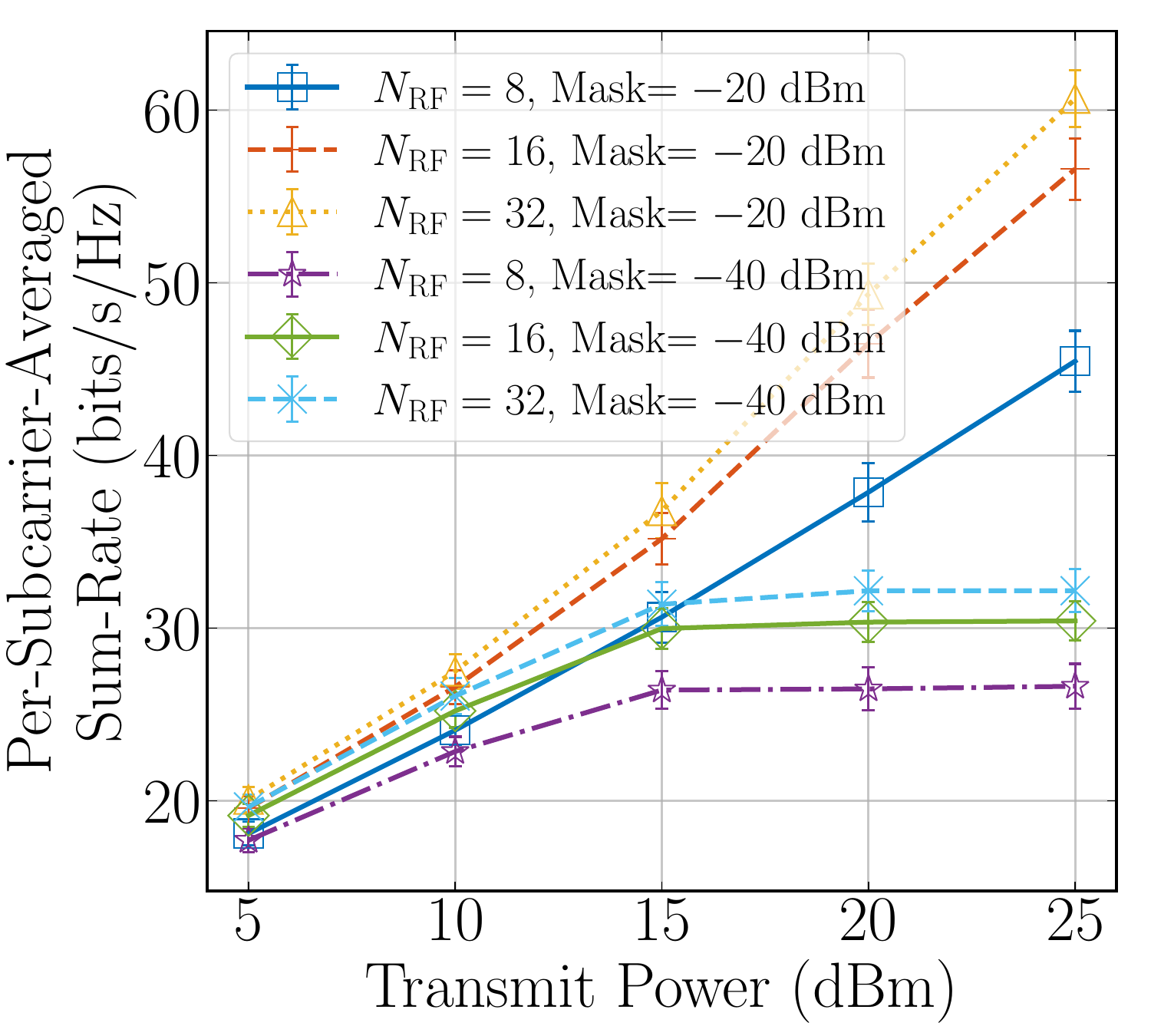}
			\captionof{figure}{Per-subcarrier-averaged sum-rate with $S = 64$.}
			\label{fig:rate_power_32}
		\end{minipage}
		\hfill
		\begin{minipage}[t]{0.315\textwidth}
			\centering
			\includegraphics[width=\textwidth]{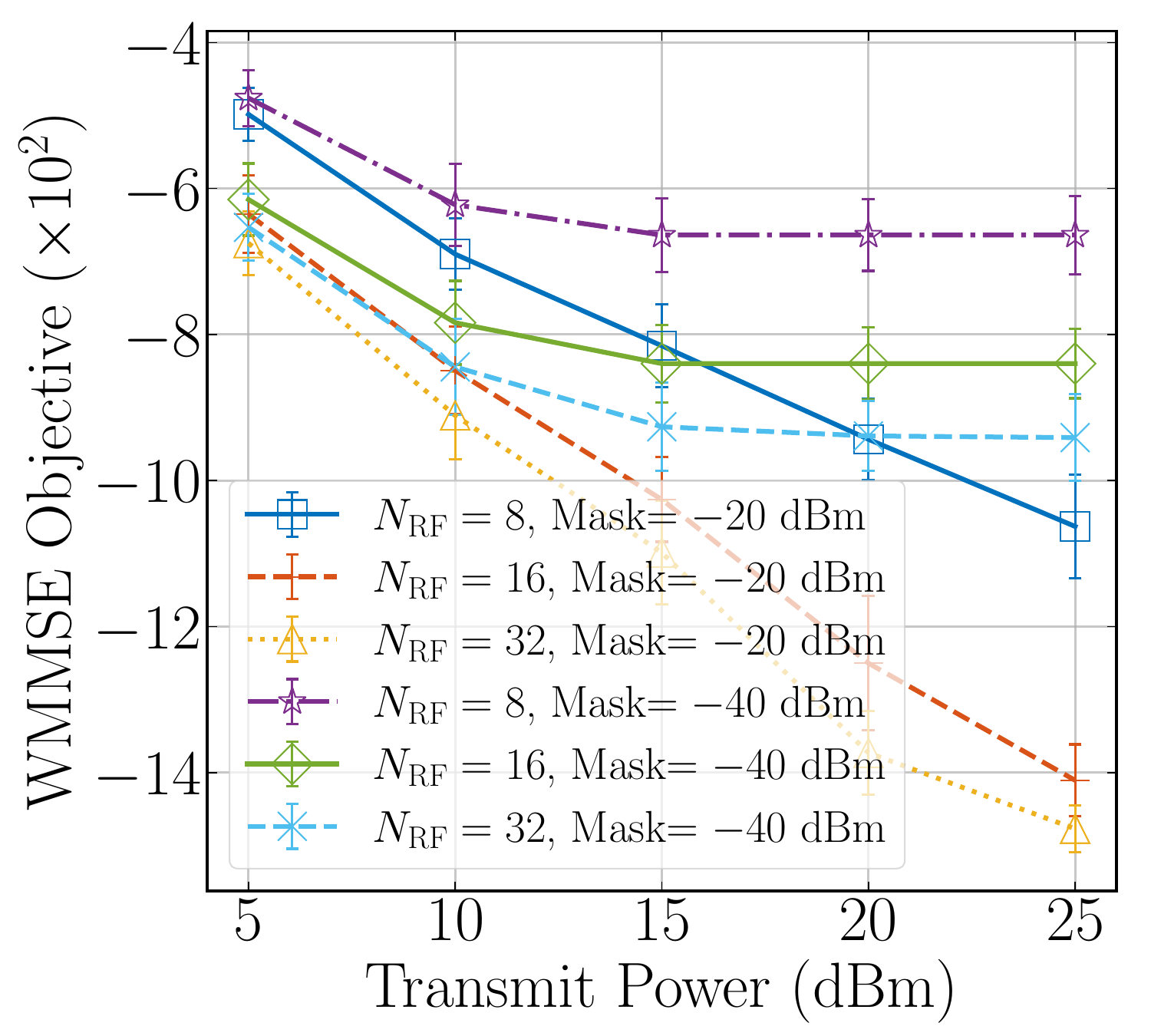}
			\captionof{figure}{Per-subcarrier-averaged WMMSE objective with $S = 64$.}
			\label{fig:rate_power_64}
		\end{minipage}
	\end{figure*}
	
			\begin{figure*}[ht!]
		\centering
		\hspace{.5cm}\includegraphics[width=1.04\textwidth]{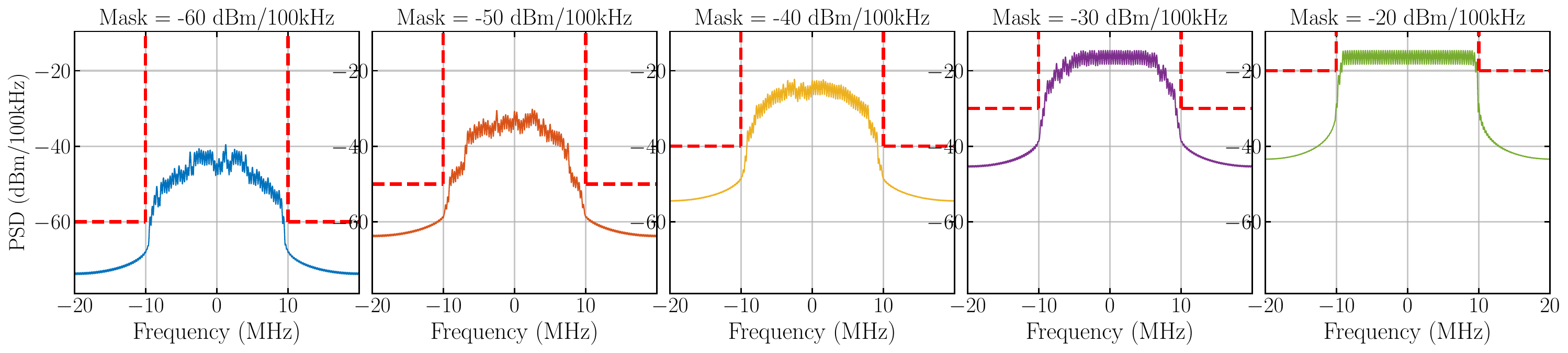}
		\caption{The expected PSD of the transmitted 20-MHz-bandwidth OFDM symbol as the considered mask changes.}
		\label{fig:psd}
	\end{figure*}
	\begin{figure}
		\centering
		\includegraphics[width=0.315\textwidth]{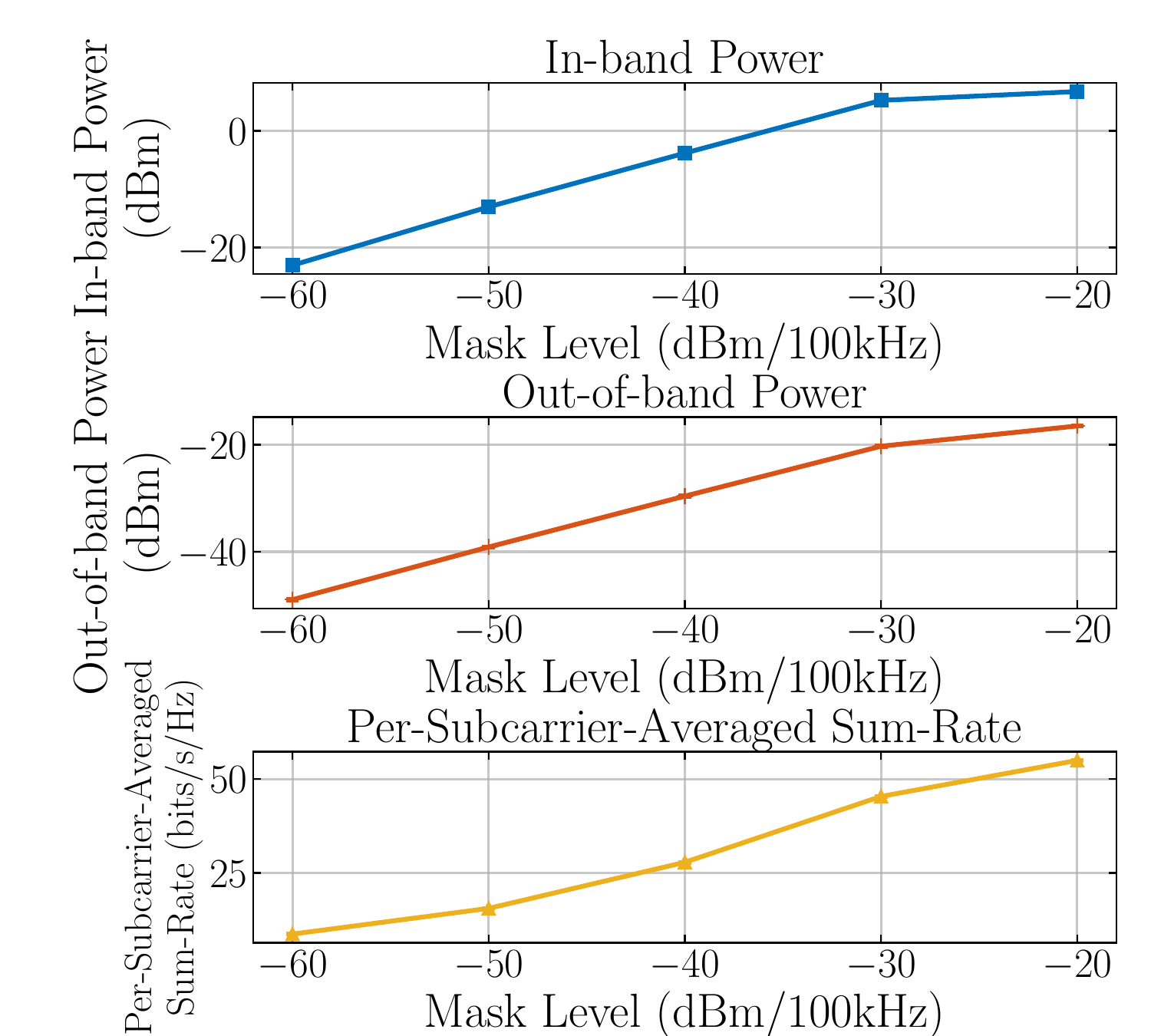}
		\caption{The in-band/out-of-band emitted power and the per-subcarrier-averaged sum-rate for different masks.}\label{fig:emission_rate}
	\end{figure}

	\subsection{The Proposed BCD-based Hybrid Precoding Approach}
	In the proposed method, different blocks are updated in a cyclic order. When we update each block, we keep the rest of the blocks fixed. In the BCD method, $\mathbf{U}_k^s$ is updated using \eqref{eq:U}. $\mathbf{W}_k^s$ is updated based on the equation given in Section \ref{sec:update_W}. We update $\mathbf{V}_k^s$ via the ADMM-based algorithm given in Section \ref{subsec:admm_blocks_hyb_provable}. The PS matrix $\mathbf{V}_{\text{RF}}$ is updated according to the method given in Section \ref{sec:analogp}. Finally, the RF combiner is updated based on the mechanism given in Section \ref{sec:updateU}.

\begin{proposition}\label{prop:bcd_convergence_hyb}
	Fix \(\eta_v>0\), let \(\Xi^q\triangleq(\{\mathbf U_k^{s,q},\mathbf W_k^{s,q},\mathbf V_k^{s,q}\}_{k,s},\mathbf V_{\rm RF}^q,\{\mathbf U_{{\rm RF},k}^q\}_k)\), and \(\Phi_{\eta_v}(\Xi)\triangleq\sum_{k,s}(\operatorname{tr}(\mathbf W_k^s\mathbf E_k^s)-\log\operatorname{det}(\mathbf W_k^s))+(\eta_v/2)\sum_{k,s}\|\mathbf V_k^s\|_F^2\), with \(\mathbf E_k^s\) evaluated at \(\Xi\). Assume the ADMM digital-precoder subproblem is solved exactly, each RF CD sweep reaches a coordinatewise fixed point of its objective, \(\sigma_{{\rm noise},k}^{s\,2}>0\), and \(\lambda_{\min}(\mathbf U_{{\rm RF},k}^H\mathbf U_{{\rm RF},k})\ge\alpha_k>0\) for all generated RF-combiner iterates. Then, \(\{\Xi^q\}\) is bounded and \(\{\Phi_{\eta_v}(\Xi^q)\}\) is non-increasing, bounded below, hence convergent. Moreover, the bound is uniform in \(\eta_v\): \(\Phi_{\eta_v}(\Xi^q)\ge\sum_{k,s}n_k(1+\log\underline e)\) with \(\underline e\) independent of \(\eta_v\), and as \(\eta_v\to0\) the regularizer \((\eta_v/2)\sum_{k,s}\|\mathbf V_k^s\|_F^2\) vanishes uniformly over the iterate set, so the limiting objective coincides with the unregularized WMMSE objective \(\Phi_0\).
\end{proposition}
\begin{proof}
	The power and unit-modulus constraints bound \(\{\mathbf V_k^s\}\), \(\mathbf V_{\rm RF}\), \(\{\mathbf U_{{\rm RF},k}\}\). In \eqref{eq:U}, the inverted matrix obeys \(\sigma_{{\rm noise},k}^{s\,2}\mathbf U_{{\rm RF},k}^H\mathbf U_{{\rm RF},k}\succeq\sigma_{{\rm noise},k}^{s\,2}\alpha_k\mathbf I\) with bounded right-hand side, so \(\{\mathbf U_k^s\}\) is bounded; by continuity on the compact closure of the generated iterates together with the noise/rank bounds, \(\mathbf E_k^s\succ\mathbf0\) uniformly over this compact set, i.e., \(\mathbf E_k^s\succeq\underline e\,\mathbf I\) for some \(\underline e>0\), hence \(\mathbf W_k^s=(\mathbf E_k^s)^{-1}\) is bounded and \(\{\Xi^q\}\) is bounded. Minimizing \(\operatorname{tr}(\mathbf W_k^s\mathbf E_k^s)-\log\operatorname{det}(\mathbf W_k^s)\) over \(\mathbf W_k^s\succ\mathbf0\) yields \(n_k+\log\operatorname{det}(\mathbf E_k^s)\ge n_k(1+\log\underline e)\); with \((\eta_v/2)\sum_{k,s}\|\mathbf V_k^s\|_F^2\ge0\), \(\Phi_{\eta_v}\) is bounded below on the iterate set, and since \(\underline e\) depends only on \(\sigma_{{\rm noise},k}^{s\,2}\), \(\alpha_k\), and the (\(\eta_v\)-independent) constraint sets, this lower bound is uniform in \(\eta_v\). Within each outer iteration, the \(\mathbf U\)- and \(\mathbf W\)-updates are exact block minimizers, the digital block is globally minimized by the exact ADMM solution of \eqref{opt:inner3block_hyb}, and each RF CD sweep is non-increasing in its objective (which equals \(\Phi_{\eta_v}\) up to terms independent of the updated block); hence \(\Phi_{\eta_v}(\Xi^{q+1})\le\Phi_{\eta_v}(\Xi^q)\). A non-increasing, bounded-below sequence converges. Finally, the iterates lie in the \(\eta_v\)-independent compact set above, on which \(\|\mathbf V_k^s\|_F^2\) is uniformly bounded; thus \((\eta_v/2)\sum_{k,s}\|\mathbf V_k^s\|_F^2\to0\) uniformly as \(\eta_v\to0\), so \(\Phi_{\eta_v}\to\Phi_0\) uniformly over the iterate set, where \(\Phi_0\) is the unregularized WMMSE objective.
\end{proof}
	
\subsection{Complexity Analysis}
\label{sec:comp}
We report per-iteration costs, separating one-time assembly from inner updates. The digital-combiner update \eqref{eq:U} costs $\mathcal O\!\big(N_rN_t+KN_{\mathrm{RF}}^2 n_k+N_r N_{\mathrm{RF}}^2+N_r^2 N_{\mathrm{RF}}+N_r^2 N_{\mathrm{RF},k}+N_{\mathrm{RF},k}^3\big)$ per $(k,s)$, dominated by the channel/interference products and the $N_{\mathrm{RF},k}\times N_{\mathrm{RF},k}$ inversion; aggregating over $(k,s)$ multiplies this by $KS$ per outer BCD iteration. The weight update $\mathbf W_k^s=(\mathbf E_k^s)^{-1}$ costs $\mathcal O(n_k^3)$.
For the digital precoder, $\boldsymbol\Psi^s$ is assembled once per ADMM call at $\mathcal O(SKN_{\mathrm{RF}}^2 n_k)$. Per ADMM iteration, forming $\{\tilde{\mathbf R}_k^{s,\tau}\}$ and the $\mathbf Z$-update each cost $\mathcal O(N_{\mathrm{RF}}SK n_k)$; the $\mathbf R$-update decomposes over RF chains, with $L_{\mathrm R}$ dual cyclic sweeps over the $G$ coordinates costing $\mathcal O(L_{\mathrm R}N_{\mathrm{RF}}GS\log(1/\delta))$ at bisection tolerance $\delta$, after which primal recovery via \eqref{eq:R_primal_hyb_provable} costs $\mathcal O(N_{\mathrm{RF}}SK n_k)$; the $\mathbf V$-update costs $\mathcal O(S(N_{\mathrm{RF}}^3+TKN_{\mathrm{RF}} n_k))$ using a one-time per-subcarrier eigendecomposition followed by $T$ bisection steps on $\vartheta^s$ over $K$ right-hand sides. The total digital-precoder cost is the per-ADMM-iteration cost times the number of ADMM iterations $L_{\rm ADMM}$.
For the RF precoder, $\mathbf Q_{\mathrm{PS}}$ is assembled once at $\mathcal O\!\big(SK(N_t^2 N_r+N_t N_r N_{\mathrm{RF},k})\big)$; one entry update in \eqref{opt:main-revisedps} costs $\mathcal O(N_t)$, so $L_{\rm PS}$ full sweeps cost $\mathcal O(L_{\rm PS}N_t^2)$. For the RF combiner, $\mathbf M_{\mathrm{RF},k}$ is assembled once at $\mathcal O\!\big(S(n_k^2 N_r^2+N_r^2 N_{\mathrm{RF},k}^2)\big)$; each entry update then costs $\mathcal O(|\texttt{idx}_k|)$, giving $\mathcal O(L_{{\rm RF},k}|\texttt{idx}_k|^2)$ over $L_{{\rm RF},k}$ sweeps.

	\section{Simulation Results}
	\label{sec:sim}

	In this section, we demonstrate the efficiency of the proposed methods via numerical simulations.
	We consider a downlink system with a transmitter having $N_t = 32$ antennas, and $K = 4$ users each equipped with $N_r = 4$ antennas and 2 RF chains.
	All four users share all subcarriers across a total bandwidth of \(20\)\,MHz, where the center frequency is $28$ GHz.
	The users are randomly distributed in a circle of radius \(4\) meters, located
	\(400\) meters from the transmitter.
The channel is modeled as a frequency-selective Rician MIMO--OFDM channel with $T$ taps. For user $k$, the channel matrix on subcarrier $s$ is given by $\mathbf{H}_k^s \in \mathbb{C}^{N_r \times N_t}$ with $\mathbf{H}_k^s = \sqrt{\frac{\kappa}{\kappa+1}} \sqrt{g_k} \mathbf{a}_r(\theta_k^{\text{AoA}}) \mathbf{a}_t^H(\theta_k^{\text{AoD}}) + \sum_{l=1}^{T-1} \sqrt{\frac{1}{\kappa+1}} \sqrt{g_{k,l}} \mathbf{a}_r(\theta_{k,l}^{\text{AoA}}) \mathbf{a}_t^H(\theta_{k,l}^{\text{AoD}}) h_{k,l} e^{-j2\pi l s/S}$ \cite{bjornson2024introduction}, where the first term represents the dominant line-of-sight (LOS) component and the summation represents the $T-1$ delayed non-line-of-sight (NLOS) components, with $h_{k,l} \sim \mathcal{CN}(0,1)$. The Rician $\kappa$-factor is defined as $\kappa = \frac{P_{\text{LOS}}}{P_{\text{NLOS}}} = 10^{\kappa_{\text{dB}}/10}$, where $\kappa_{\text{dB}} = 10$~dB.
	The path loss coefficients follow the 3GPP path loss model \cite{3GPP_TR_36_814_2017}. For the LOS component, we have $
		PL_{\text{dB}}^{\text{LOS}} = 22 \log_{10}(d_k) + 28 + 20 \log_{10}(f^c)+ \xi_{\text{LOS}},$
	where $d_k$ is the distance in meters, $f^c$ is the carrier frequency in GHz, $\xi_{\text{LOS}} \sim \mathcal{N}(0, \sigma_{\text{LOS}}^2)$ represents log-normal shadow fading with $\sigma_{\text{LOS}} = 5.8$ dB, and $g_k = 10^{-PL_{\text{dB}}^{\text{LOS}}/10}$. For the NLOS components $
		PL_{\text{dB}}^{\text{NLOS}} = 22 \log_{10}(d_k) + 28 + 20 \log_{10}(f^c)+ \xi_{\text{NLOS}},$
	where $\xi_{\text{NLOS}} \sim \mathcal{N}(0, \sigma_{\text{NLOS}}^2)$ with $\sigma_{\text{NLOS}} = 8.7$ dB, and $g_{k,l} = 10^{-PL_{\text{dB}}^{\text{NLOS}}/10}$ represents the channel gain.
	
	The additive noise is modeled as spatially and temporally white complex Gaussian noise $
		\mathbf{n}_k^s \sim \mathcal{CN}(\mathbf{0}, \sigma_{\text{noise},k}^{s\,2} \mathbf{I}_{N_r}).$ 
	The noise power spectral density is $-174$ dBm/Hz. The noise figure of each user device is $8$ dB.

Uniform linear arrays (ULAs) are assumed at both the transmitter and receivers, with responses
$\mathbf{a}_t(\theta)=\big[1,e^{\jmath\psi},e^{\jmath2\psi},\ldots,e^{\jmath(N_t-1)\psi}\big]^T$ and
$\mathbf{a}_r(\phi)=\big[1,e^{\jmath\psi'},e^{\jmath2\psi'},\ldots,e^{\jmath(N_r-1)\psi'}\big]^T$, where
$\psi=\frac{2\pi d_t\sin(\theta)}{\lambda}$, $\psi'=\frac{2\pi d_r\sin(\phi)}{\lambda}$,
and $\theta$, $\phi$ are measured from each array's broadside.
We set $d_t=d_r=\lambda/2$; at $28$~GHz, $\lambda\approx 10.7$~mm.
Each user's receive array is oriented with its broadside facing the transmitter,
and the LOS angle of departure $\theta_k^{\text{AoD}}$ is determined by
user~$k$'s geometric position relative to the transmitter array broadside.
	
	For NLOS components, the angles are randomly distributed around the main LOS angles with angular spread $\sigma_\theta$ as
	$
		\theta_{k,l}^{\text{AoD}} = \theta_k^{\text{AoD}} + \Delta\theta_{k,l},
		\theta_{k,l}^{\text{AoA}} = \theta_k^{\text{AoA}} + \Delta\phi_{k,l},$
	where $\Delta\theta_{k,l}, \Delta\phi_{k,l} \sim \mathcal{N}(0, \sigma_\theta^2)$.

	The maximum allowed amplitude of the unclipped signal is \(\chi = 0.7\), 
	and the clipping and spectral mask constraints are satisfied with a probability 
	of at least \(90\%\). To avoid OOB emissions, 
	we set the spectral mask constraint frequencies from $-20$ MHz to $-10.01$ MHz and $10.01$ MHz to $20$ MHz using 90 uniformly spaced points per side, where the flat limits are ${-20,-40}$~dBm/$100$~kHz. With \(64\) subcarriers, we plot the per-subcarrier convergence of the proposed 
	BCD algorithm for different numbers of RF chains in 
	Fig.~\ref{fig:conv_bcd}, where the per-subcarrier power budget is $25$ dBm. We observe that the algorithm’s per-subcarrier-averaged sum-rate increases monotonically with each iteration. The per-subcarrier-averaged sum-rate and WMMSE cost function are shown  in Figs.~\ref{fig:rate_power_32} and	\ref{fig:rate_power_64}, respectively. The per-subcarrier-averaged sum-rate increases with higher 
	transmit power budgets and a greater number of RF chains. The per-subcarrier-averaged sum-rate increases markedly at low power, but the rate of improvement diminishes at higher power as the spectral-mask constraints become active and limit further increment.

	\begin{figure}
		\centering
		\includegraphics[width=.42\textwidth]{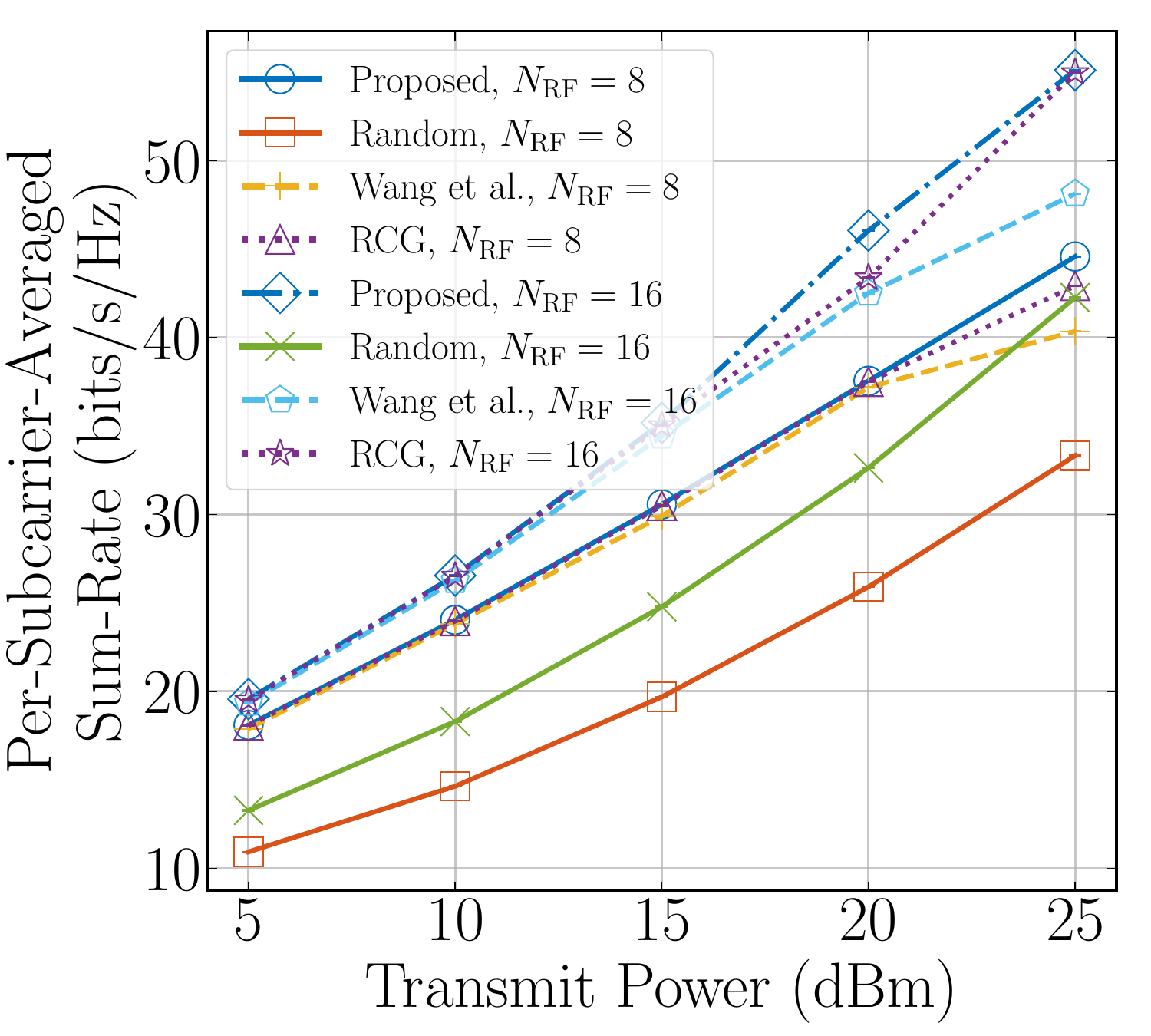}
		\caption{The per-subcarrier-averaged sum-rate of different hybrid precoding methods.}\label{fig:analog_rate}
	\end{figure}

	Next, we demonstrate how OOB emissions vary with spectral masks. The mask starts at $|f|=10.01$~MHz and enforces a flat limit for $10.01\leq|f|$~MHz, where the flat limit is swept across five levels ${-20,-30,-40,-50,-60}$~dBm/$100$~kHz. We show the expected PSD 
	of the transmitted 20\,MHz OFDM symbol in Fig.~\ref{fig:psd}, 
	while varying the enforced spectral mask. 
	It is observed that as the enforced spectral mask becomes tighter, 
	both the allowable transmit power and the width of the expected OFDM symbol's PSD decrease. 

The in-band emissions, OOB emissions, and the per-subcarrier-averaged sum-rate are depicted in Fig.~\ref{fig:emission_rate}, where $P^s = 25$\,dBm and $N_{\text{RF}} = 16$. The optimized hybrid precoding results in $-23$\,dBm of in-band emissions and $-48.9$\,dBm of OOB emissions when the sidelobes of the mask is $-60$~dBm/$100$~kHz. The per-subcarrier-averaged sum-rate is \(8.69\)\,bps/Hz. When the mask is increased to $-50$~dBm/$100$~kHz, the in-band emissions rise to \(-13\)\,dBm, and the OOB emissions increase to \(-39.1\)\,dBm. Consequently, the per-subcarrier-averaged sum-rate increases to \(15.54\)\,bps/Hz. Therefore, increasing the mask limit results in higher OOB emissions and an improved per-subcarrier-averaged sum-rate. Additional data points are shown in Fig.~\ref{fig:emission_rate} as the allowable limit is incrementally increased.
	
		\begin{figure}
		\centering
		\includegraphics[width=.42\textwidth]{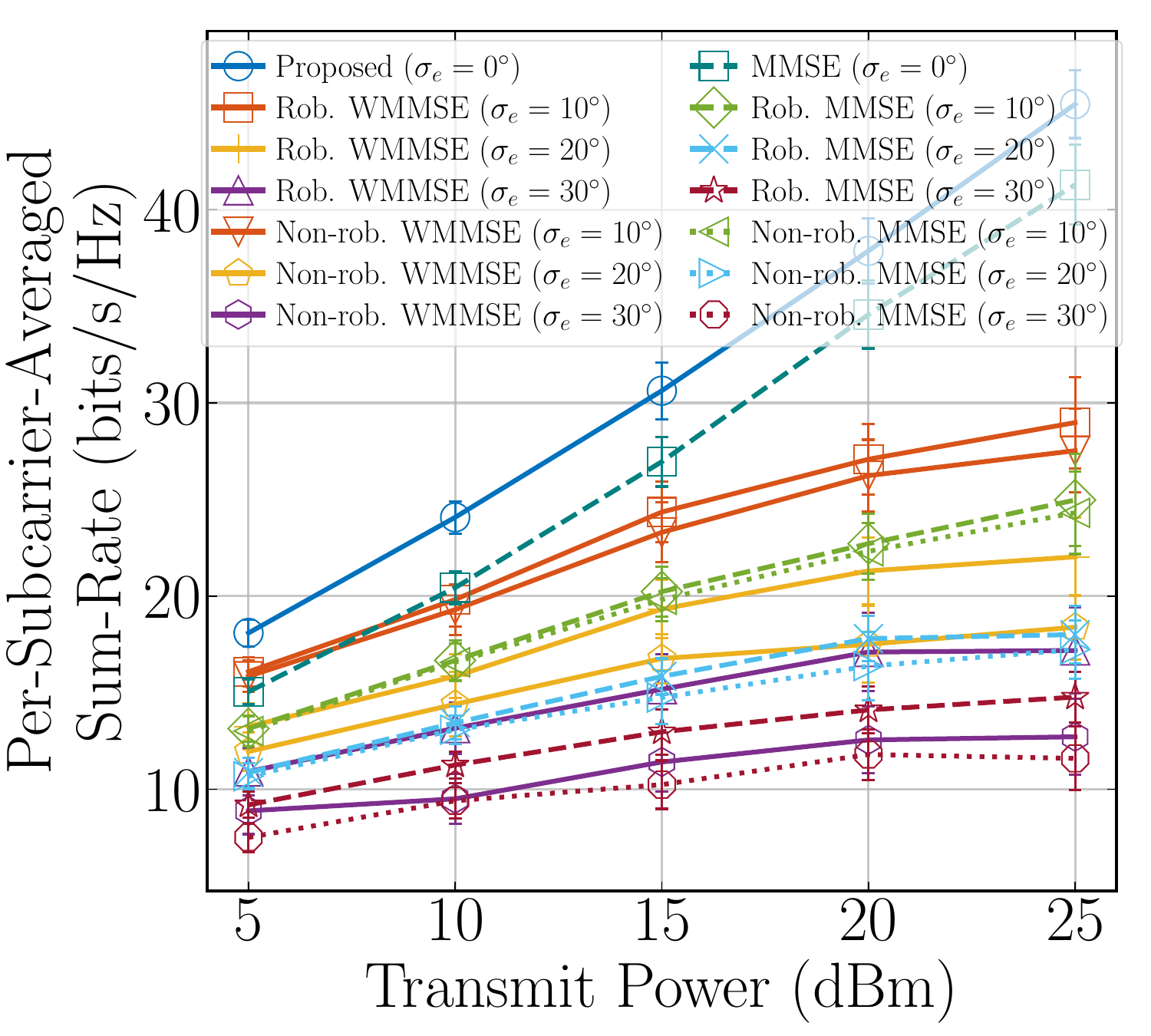}
		\caption{The classical hybrid MMSE and the proposed BCD methods with different PS noise levels when $N_{\text{RF}}=8$.}\label{fig:robust1}
	\end{figure}
	We compare our proposed RF precoding-combination scheme with the method presented by Wang \textit{et al.} \cite{wang2018hybrid}. This reference employs a CD algorithm that iteratively minimizes the objective function. In this algorithm, individual PSs at the transmitter or for a specific user are alternately optimized via a numerical search method, while all other PS values are held constant. We also demonstrate the performance of the RCG method studied in \cite{7397861,9467491} in Fig. \ref{fig:analog_rate}. The random assignment of PS values for the precoder and combiners is additionally considered. As illustrated in Fig. \ref{fig:analog_rate}, which depicts the performance across varying numbers of RF chains, the proposed method outperforms these benchmark approaches in terms of per-subcarrier-averaged sum-rate.

	When $\mathbf{W}_k^s= \mathbf{I}_{n_k}$, the proposed method simplifies to the hybrid MMSE precoding/combining method. With ideal noise-free PSs ($\sigma_e=0^{\circ}$), we compare the performance of the proposed and MMSE methods in Fig. \ref{fig:robust1} for $N_{\text{RF}}=8$. It is observed that the proposed algorithm achieves higher per-subcarrier-averaged sum-rates compared to the hybrid MMSE method. Different levels of noise strength is considered $\sigma_e=[10^{\circ},20^{\circ},30^{\circ}]$. From Fig. \ref{fig:robust1}, it is observed that phase noise of PSs degrade the per-subcarrier-averaged sum-rate of the system. When the proposed robust PS update rules are deployed, the given algorithms (either proposed or MMSE) outperform non-robust versions (with PS update rules not accounting for PS errors) in the presence of PS impairments.

	\section{Concluding Remarks}
	\label{sec:con}
	
	In this paper, we studied the block-level hybrid digital-RF precoding design for MU-MIMO-OFDM systems with possibly impaired PSs subject to constraints on OOB emissions and clipping. We proposed a WMMSE-based BCD method to maximize the sum-rate of users. A problem decomposition and an ADMM approach were proposed to optimize the digital precoder of the transmitter, where the solutions for the ADMM subproblems were obtained through closed-forms and bisection searches. CD methods with closed-form updates were proposed to optimize the phase shifts at the transmitter and users, where the unit modulus constraints are considered for PSs. Optimizing with respect to different blocks with the Gauss-Seidel rule is continued until the objective value converges. We demonstrated the efficiency of the proposed approaches against classical methods via extensive simulations.

	\ifCLASSOPTIONcaptionsoff
	\newpage
	\fi

	\bibliographystyle{IEEEbib}
	\bibliography{ref_SDRA}

@article{tao2018convergence,
		title={Convergence analysis of the direct extension of {ADMM} for multiple-block separable convex minimization},
		author={M. Tao and X. Yuan},
		journal={Adv. Comput. Math.},
		volume={44},
		number={3},
		pages={773--813},
		year={Jun. 2018},
		publisher={Springer}
	}

@article{bjornson2024introduction,
		title={Introduction to multiple antenna communications and reconfigurable surfaces},
		author={E. Bj{\"o}rnson and {\"O}. Demir},
		year={2024},
		publisher={Now Publishers, Inc.}
	}

@techreport{3GPP_TS_38_104,
		author      = {{3GPP TS 38.104}},
		title       = {{{5G} {NR}; Base Station ({BS}) radio transmission and reception}},
		institution = {3rd Generation Partnership Project (3GPP)},
		version     = {V16.4.0},
		year        = {Jul. 2020},
	}

@techreport{3GPP_TR_36_814_2017,
		author    = {{3GPP TR 36.814}},
		title     = {Further advancements for {E-UTRA} physical layer aspects},
		institution = {3rd Generation Partnership Project (3GPP)},
		year      = {2017},
		month     = {Mar.},
		note      = {Release 9, V9.2.0}
	}

@ARTICLE{4014335,
		author={R. Baxley and C. Zhao and G.T. Zhou},
		journal={IEEE Trans. Broadcast.}, 
		title={Constrained Clipping for Crest Factor Reduction in {OFDM}}, 
		year={Dec. 2006},
		volume={52},
		number={4},
		pages={570-575},
		doi={10.1109/TBC.2006.883301}}

@ARTICLE{1214054,
			author={C. Hangjun and A.M. Haimovich},
			journal={IEEE Commun. Lett.}, 
			title={Iterative estimation and cancellation of clipping noise for {OFDM} signals}, 
			year={Jul. 2003},
			volume={7},
			number={7},
			pages={305-307},
			doi={10.1109/LCOMM.2003.814720}}

@ARTICLE{4786507,
		author={U.-K. Kwon and D. Kim and G.-H. Im},
		journal={IEEE Trans. Wireless Commun.}, 
		title={Amplitude clipping and iterative reconstruction of {MIMO}-{OFDM} signals with optimum equalization}, 
		year={Jan. 2009},
		volume={8},
		number={1},
		pages={268-277},
		doi={10.1109/T-WC.2009.071034}}

@ARTICLE{7956180,
		author={R. Garg and A. S. Natarajan},
		journal={IEEE Trans. Microw. Theory Tech.}, 
		title={A 28-{GHz} Low-Power Phased-Array Receiver Front-End With 360$^\circ$ {RTPS} Phase Shift Range}, 
		year={Jun. 2017},
		volume={65},
		number={11},
		pages={4703-4714},
		doi={10.1109/TMTT.2017.2707414}}

@ARTICLE{7948784,
			author={C.W. Byeon and C.S. Park},
			journal={IEEE Microw. Wireless Compon. Lett}, 
			title={A Low-Loss Compact 60-{GHz} Phase Shifter in 65-nm {CMOS}}, 
			year={Jun. 2017},
			volume={27},
			number={7},
			pages={663-665},
			doi={10.1109/LMWC.2017.2711569}}

@ARTICLE{9767793,
		author={J. Li and Z. Wang and Y. Zhang and P. Zhu and D. Wang and X. You},
		journal={IEEE Syst. J.}, 
		title={Robust Hybrid Beamforming for Outage-Constrained Multigroup Multicast {mmWave} Transmission With Phase Shifter Impairments}, 
		year={Mar. 2023},
		volume={17},
		number={1},
		pages={869-880},
		doi={10.1109/JSYST.2022.3168022}}

@ARTICLE{9174747,
	author={S. Gong and C. Xing and V. Lau and S. Chen and L. Hanzo},
	journal={IEEE Trans. Signal Proces.}, 
	title={Majorization-Minimization Aided Hybrid Transceivers for {MIMO} Interference Channels}, 
	year={Aug. 2020},
	volume={68},
	number={},
	pages={4903-4918}}

@ARTICLE{10778226,
		author={M. Alouzi and H. Yanikomeroglu and G. Karabulut Kurt},
		journal={IEEE Trans. Wireless Commun.}, 
		title={Adaptive Phase Shifters for Hybrid Beamforming in {mmWave} Systems}, 
		year={Feb. 2025},
		volume={24},
		number={2},
		pages={1104-1116},
		doi={10.1109/TWC.2024.3505202}}

@ARTICLE{10972245,
		author={H. Vaezy and S.D. Blostein},
		journal={IEEE Trans. Signal Process.}, 
		title={Joint User Selection and Hybrid Precoder Design for Massive {MIMO} Systems}, 
		year={Apr. 2025},
		volume={73},
		number={},
		pages={1808-1822},
		doi={10.1109/TSP.2025.3562858}}

@ARTICLE{10839437,
		title={Hybrid Precoding for {mmWave} Massive {MIMO} With Finite Blocklength},
		author={X. Zhang and L. Xiang and J. Wang and P. Zhu and D.W.K. Ng and X. Gao},
		journal={IEEE Trans. Commun.}, 
		year={Aug. 2025},
		volume={73},
		number={8},
		pages={6379-6395},
		doi={10.1109/TCOMM.2025.3529244}}

@ARTICLE{8048606,
	author={D. Nguyen and L. Le Bao and T.  Le-Ngoc and R.W. Heath},
	journal={IEEE Access}, 
	title={Hybrid {MMSE} Precoding and Combining Designs for {mmWave} Multiuser Systems}, 
	year={Sep. 2017},
	volume={5},
	number={},
	pages={19167-19181},
	doi={10.1109/ACCESS.2017.2754979}}

@ARTICLE{5628256,
	author={Y.-C. Wang and Z.-Q. Luo},
	journal={IEEE Trans. Commun.}, 
	title={Optimized Iterative Clipping and Filtering for {PAPR} Reduction of {OFDM} Signals}, 
	year={Jan. 2011},
	volume={59},
	number={1},
	pages={33-37},
	doi={10.1109/TCOMM.2010.102910.090040}}

@ARTICLE{9151131,
	author={J. Galaviz-Aguilar V. Alejandro and T. Cesar and E. Tlelo-Cuautle},
	journal={IEEE Access}, 
	title={{RF-PA} Modeling of {PAPR}: A Precomputed Approach to Reinforce Spectral Efficiency}, 
	year={Jul. 2020},
	volume={8},
	number={},
	pages={138217-138235},
	doi={10.1109/ACCESS.2020.3012610}}

@ARTICLE{10048700,
		author={S. Liu and Y. Wang and Z. Lian and Y. Su and Z. Xie},
		journal={IEEE Trans. Broadcasting}, 
		title={Joint Suppression of {PAPR} and {OOB} Radiation for {OFDM} Systems}, 
		year={Feb. 2023},
		volume={69},
		number={2},
		pages={528-537},
		doi={10.1109/TBC.2023.3243410}}

@article{tseng2001convergence,
			author  = {P. Tseng},
			title   = {Convergence of a Block Coordinate Descent Method for Nondifferentiable Minimization},
			journal = {J. Optim. Theory Appl.},
			volume  = {109},
			number  = {3},
			pages   = {475--494},
			year    = {Jun. 2001}
		}

@ARTICLE{9467491,
	author={X. Zhao and T. Lin and Y. Zhu and J. Zhang},
	journal={IEEE Trans. Wireless Commun.}, 
	title={Partially-Connected Hybrid Beamforming for Spectral Efficiency Maximization via a Weighted {MMSE} Equivalence}, 
	year={Jul. 2021},
	volume={20},
	number={12},
	pages={8218-8232},
	doi={10.1109/TWC.2021.3091524}}

@ARTICLE{1468466,
		author={M. Joham and W. Utschick and J.A. Nossek},
		journal={IEEE Trans. Signal Process.}, 
		title={Linear transmit processing in {MIMO} communications systems}, 
		year={Aug. 2005},
		volume={53},
		number={8},
		pages={2700-2712},
		doi={10.1109/TSP.2005.850331}}

@article{van2009sculpting,
		title={Sculpting the multicarrier spectrum: a novel projection precoder},
		author={J. Van De Beek},
		journal={IEEE Commun. Lett.},
		volume={13},
		number={12},
		pages={881--883},
		year={Dec. 2009},
		publisher={IEEE}
	}

@INPROCEEDINGS{icc2026,
			author={N. Reyhanian and R.G. Zefreh, P. Ramezani and E. Bj{\"o}rnson},
			booktitle={Proc. IEEE Int. Conf. Commun. (ICC)}, 
			title={Hybrid Precoding for Multi-User MIMO-OFDM Systems with Phase Shifter Impairments}, 
			year={May 2026},
			volume={},
			number={},
			doi={10.1109/PIMRC.2018.8580998}}

@ARTICLE{9390405,
		author={S. Kant and M. Bengtsson and B. Göransson and G. Fodor and C. Fischione},
		journal={IEEE Trans. Wireless Commun.}, 
		title={Efficient Optimization for Large-Scale {MIMO}-{OFDM} Spectral Precoding}, 
		year={Sep. 2021},
		volume={20},
		number={9},
		pages={5496-5513},
		doi={10.1109/TWC.2021.3068207}}

@ARTICLE{8653302,
	author={D. Mishra and H. Johansson},
	journal={IEEE Trans. Commun.}, 
	title={Optimal Channel Estimation for Hybrid Energy Beamforming Under Phase Shifter Impairments}, 
	year={Jun. 2019},
	volume={67},
	number={6},
	pages={4309-4325},
	doi={10.1109/TCOMM.2019.2901790}}

@ARTICLE{6459499,
	author={A. Tom and A. Sahin and H. Arslan},
	journal={IEEE Commun. Lett.}, 
	title={Mask Compliant Precoder for {OFDM} Spectrum Shaping}, 
	year={Mar. 2013},
	volume={17},
	number={3},
	pages={447-450},
	doi={10.1109/LCOMM.2013.020513.122495}}

@ARTICLE{7485853,
	author={R. Kumar and A. Tyagi},
	journal={IEEE Trans. Cognit. Commun. Netw.}, 
	title={Computationally Efficient Mask-Compliant Spectral Precoder for {OFDM} Cognitive Radio}, 
	year={Mar. 2016},
	volume={2},
	number={1},
	pages={15-23},
	doi={10.1109/TCCN.2016.2577039}}

@ARTICLE{9214877,
	author={S. Kant and M. Bengtsson and G. Fodor and B. Göransson and C. Fischione, Carlo},
	journal={IEEE Trans. Wireless Commun.}, 
	title={{EVM}-Constrained and Mask-Compliant {MIMO}-{OFDM} Spectral Precoding}, 
	year={Jan. 2021},
	volume={20},
	number={1},
	pages={590-606},
	doi={10.1109/TWC.2020.3027345}}

@article{wang2018hybrid,
	title={Hybrid precoder and combiner design with low-resolution phase shifters in {mmWave} {MIMO} systems},
	author={Z. {Wang} and M. {Li} and Q. {Liu} and A. L. {Swindlehurst}},
	journal={IEEE J. Sel. Topics Signal Process.},
	volume={12},
	pages={256--269},
	year={May 2018}
}

@article{taheri2020joint,
	title={Joint spectral-spatial precoders in {MIMO}-{OFDM} transmitters},
	author={T. Taheri and M. Mohamad and R. Nilsson and J. van de Beek},
	journal={Signal Process.},
	volume={172},
	pages={107538},
	year={Jul. 2020},
	publisher={Elsevier}
}

@ARTICLE{7397861,
	author={X. Yu and J.C. Shen and J. Zhang and K.B. Letaief},
	journal={IEEE J. Sel. Topics Signal Proces.}, 
	title={Alternating Minimization Algorithms for Hybrid Precoding in Millimeter Wave {MIMO} Systems}, 
	year={Feb. 2016},
	volume={10},
	number={3},
	pages={485-500},
	doi={10.1109/JSTSP.2016.2523903}}

@ARTICLE{7913599,
		author={F. Sohrabi and W. {Yu}},
		journal={IEEE J. Sel. Areas Commun.}, 
		title={Hybrid Analog and Digital Beamforming for {mmWave} {OFDM} Large-Scale Antenna Arrays}, 
		year={Apr. 2017},
		volume={35},
		number={7},
		pages={1432-1443},
		doi={10.1109/JSAC.2017.2698958}}

@ARTICLE{5756489,
	author={Q. Shi \textit{et al.}},
	journal={IEEE Trans. Signal Process.}, 
	title={An Iteratively Weighted {MMSE} Approach to Distributed Sum-Utility Maximization for a {MIMO} Interfering Broadcast Channel}, 
	year={Sep. 2011},
	volume={59},
	number={9},
	pages={4331-4340}}

@ARTICLE{7366560,
	author={A. Tom and A. Şahin and H. Arslan},
	journal={IEEE Trans. Commun.}, 
	title={Suppressing Alignment: Joint {PAPR} and Out-of-Band Power Leakage Reduction for {OFDM}-Based Systems}, 
	year={Dec. 2016},
	volume={64},
	number={3},
	pages={1100-1109},
	doi={10.1109/TCOMM.2015.2512603}}

@ARTICLE{6928432,
	author={L. Liang and W. Xu and X. Dong},
	journal={IEEE Wireless Commun. Lett.}, 
	title={Low-Complexity Hybrid Precoding in Massive Multiuser {MIMO} Systems}, 
	year={Oct. 2014},
	volume={3},
	number={6},
	pages={653-656},
	doi={10.1109/LWC.2014.2363831}}

@book{bertsekas1999nonlinear,
	title={Nonlinear Programming},
	author={D. P. {Bertsekas}},
	isbn={9781886529052},
	year={3rd ed., 2016},
	publisher={Athena Scientific}
}
	\newpage

\end{document}